\documentclass[11pt,letterpaper]{article}

\usepackage[authoryear]{natbib}
\usepackage[breaklinks]{hyperref}

\usepackage[ruled]{algorithm2e}
\renewcommand{\algorithmcfname}{ALGORITHM}
\SetAlFnt{\small}
\SetAlCapFnt{\small}
\SetAlCapNameFnt{\small}
\SetAlCapHSkip{0pt}
\IncMargin{-\parindent}

\usepackage{fullpage}
\usepackage{mysetup,my-acm} % Alex's latex setup

\usepackage{color,xspace}
\usepackage{graphicx}

\usepackage{algorithmic}
\usepackage{subfigure}

\usepackage{color-edits}
\definecolor{darkgreen}{rgb}{0,0.5,0}
\addauthor{JWV}{magenta}  % jw for Jenn
\addauthor{as}{red}   % as for Alex
\addauthor{cj}{darkgreen} % cj for CJ
\renewcommand{\refeq}[1]{Equation~(\ref{#1})}
\newcommand{\Exp}[2][]{{\operatornamewithlimits{\mathbb{E}}_{#1}\left[{#2}\right]}}

\newcommand{\problem}{dynamic contract design\xspace} % A.S.: we can change this later
\newcommand{\Problem}{Dynamic contract design\xspace}
\newcommand{\taskPricing}{dynamic task pricing\xspace} % A.S.: we can change this later
\newcommand{\TaskPricing}{Dynamic task pricing\xspace} % A.S.: we can change this later
\newcommand{\incr}{\mathtt{incr}}    % increment representation
\newcommand{\width}{\mathtt{width}}  % width
\newcommand{\vw}{\mathtt{VirtWidth}}   % virtual width
\newcommand{\confC}{c_{\mathtt{rad}}} % the constant in the confidence radius
\newcommand{\Xcand}{X_{\mathtt{cand}}} % candidate contracts
\newcommand{\OPT}{\mathtt{OPT}}
\newcommand{\WidthDim}{\mathtt{WidthDim}} % width dimension
\newcommand{\ZoomDim}{\mathtt{ZoomDim}} % zooming dimension
\newcommand{\minSize}{\psi} % min size of an active cell
\newcommand{\SlackC}{\beta_{0}} % the slack constant in the main theorem
\newcommand{\Nmin}{N^{\mathtt{min}}} % covering numbers in terms of eps-minimal cells
\newcommand{\Mesh}{\Xcand}

\newcommand{\hlexample}{high-low example\xspace}

\newcommand{\tildeO}{\tilde{O}}
\newcommand{\mP}{\mathcal{P}}

\newcommand{\zooming}{\ensuremath{\mathtt{AgnosticZooming}}\xspace} % Agnostic Zooming
\newcommand{\UCB}{\ensuremath{\mathtt{UCB1}}\xspace} % UCB1
\newcommand{\NonAdaptive}{\ensuremath{\mathtt{NonAdaptive}}\xspace} % Non-adaptive discretization

\newcommand{\ignore}[1]{}
\newcommand{\cj}[1]{\cjcomment{#1}}
\newcommand{\jenn}[1]{\JWVcomment{#1}}

\newcommand{\rad}{\mathtt{rad}}

\renewcommand{\vec}[1]{\mathbf{#1}}
\newcommand{\abs}[1]{\left\lvert#1\right\rvert}

\newtheorem{dfn}[theorem]{Definition}

\title{Adaptive Contract Design for Crowdsourcing Markets:\\
Bandit Algorithms for Repeated Principal-Agent Problems%
\thanks{Preliminary version of this paper has been published in the \emph{ACM Conference on Economics and Computation (ACM-EC)}, 2014. The conference version omits many of the proofs, including all of Section~\ref{sec:pricing}, the details of the simulation results, and a detailed discussion of related work (Section~\ref{sec:related-work}). Moreover, the present version contains revised and expanded Conclusions section, and an updated discussion of the follow-up work.
\vspace{2mm}\newline Compared to the initial technical report ({\tt arXiv:1405.2875v1}, May 2014), this version includes updated citations and discussion of the follow-up work.\vspace{2mm}\newline 
Much of this research was completed while Ho was an intern at Microsoft Research. This research was partially supported by the NSF under grant IIS-1054911. Any opinions, findings, conclusions, or recommendations are those of the authors alone. }
}

\author{Chien-Ju Ho%
\footnote{UCLA, Los Angeles, CA, USA. Email: {\tt cjho@ucla.edu}.}
\and
Aleksandrs Slivkins%
\footnote{Microsoft Research, New York, NY, USA. Email: {\tt slivkins@microsoft.com}.}
\and
Jennifer Wortman Vaughan%
\footnote{Microsoft Research, New York, NY, USA. Email: {\tt jenn@microsoft.com}.}
}

\date{First version: May 2014 \\ This version: July 2015.}

\begin{document}

\maketitle

\begin{abstract}
Crowdsourcing markets have emerged as a popular platform for matching available workers with tasks to complete.  The payment for a particular task is typically set by the task's requester, and may be adjusted based on the quality of the completed work, for example, through the use of ``bonus'' payments.  In this paper, we study the requester's problem of dynamically adjusting quality-contingent payments for tasks. We consider a multi-round version of the well-known \emph{principal-agent} model, whereby in each round a worker makes a strategic choice of the effort level which is not directly observable by the requester. In particular, our formulation significantly generalizes the budget-free online task pricing problems studied in prior work.

We treat this problem as a multi-armed bandit problem, with each ``arm" representing a potential contract. To cope with the large (and in fact, infinite) number of arms, we propose a new algorithm, \zooming, which discretizes the contract space into a finite number of regions, effectively treating each region as a single arm.  This discretization is adaptively refined, so that more promising regions of the contract space are eventually discretized more finely. We analyze this algorithm, showing that it achieves regret sublinear in the time horizon and substantially improves over non-adaptive discretization (which is the only competing approach in the literature).

Our results advance the state of art on several different topics: the theory of crowdsourcing markets, principal-agent problems, multi-armed bandits, and dynamic pricing.

\ignore{
In our setting, each round the requester posts a new contract, a performance-contingent payment rule which specifies different levels of payment for different qualities of work.  A random, unidentifiable worker then arrives in the market and decides whether to accept the requester's task and how much effort to exert.  After the worker completes the task (or chooses not to complete it), the requester observes the quality of the work produced, pays the worker according to the offered contract, and adjusts the contract for the next round.  The goal of the requester is to maximize his own expected utility, the value he receives from completed work minus the payments made.

 To analyze the algorithm's performance, we propose a concept called ``width dimension'' which measures how ``nice'' a particular problem instance is. We show that Agnostic Zooming achieves regret bounds sub-linear in the time horizon for problem instances with small width dimension. Even for problem instances with large width dimension, Agnostic Zooming can match the performance of the naive algorithm which uniformly discretizes the space and runs a standard bandit algorithm.
} 
\end{abstract}

{\bf ACM Categories and subject descriptors:}


\newpage

\section{Introduction}
\label{sec:intro}
Crowdsourcing harnesses human intelligence and common sense to complete tasks that are difficult to accomplish using computers alone.  Crowdsourcing markets, such as Amazon Mechanical Turk and Microsoft's Universal Human Relevance System, are platforms designed to match available human workers with tasks to complete.  Using these platforms, requesters may post tasks that they would like completed, along with the amount of money they are willing to pay.  Workers then choose whether or not to accept the available tasks and complete the work.

Of course not all human workers are equal, nor is all human-produced work.  Some tasks, such as proofreading English text, are easier for some workers than others, requiring less effort to produce high quality results.  Additionally, some workers are more dedicated than others, willing to spend extra time to make sure a task is completed properly.  To encourage high quality results, requesters may set quality-contingent ``bonus'' payments on top of the base payment for each task, rewarding workers for producing valuable output.  This can be viewed as offering workers a ``contract'' that specifies how much they will be paid based on the quality of their output.\footnote{For some tasks, such as labeling websites as relevant to a particular search query or not, verifying the quality of work may be as difficult as completing the task.  These tasks can be assigned in batches, with each batch containing one or more instances in which the correct answer is already known.  Quality-contingent payments can then be based on the known instances.}

We examine the requester's problem of dynamically setting quality-contingent payments for tasks. We consider a setting in which time evolves in rounds.  In each round, the requester posts a new contract, a performance-contingent payment rule which specifies different levels of payment for different levels of output.  A random, unidentifiable worker then arrives in the market and strategically decides whether to accept the requester's task and how much effort to exert; the choice of effort level is not directly observable by the requester.  After the worker completes the task (or chooses not to complete it), the requester observes the worker's output, pays the worker according to the offered contract, and adjusts the contract for the next round. The properties of a random worker (formally: the distribution over the workers' types) are not known to the requester, but may be learned over time. The goal of the requester is to maximize his expected utility, the value he receives from completed work minus the payments made.
We call it the \emph{\problem} problem.

For concreteness, consider a special case in which a worker can strategically choose to perform a task with low effort or with high effort, and the task may be completed either at low quality or at high quality. The low effort incurs no cost and results in low quality, which in turn brings no value to the requester. The high effort leads to high quality with some positive probability (which may vary from one worker to another, and is unknown to the requester). The requester only observes the quality of completed tasks, and therefore cannot infer the effort level. This example captures the two main tenets of our model: that the properties of a random worker are unknown to the requester and that  workers' strategic decisions are unobservable.

We treat the \problem problem as a multi-armed bandit (MAB) problem, with each arm representing a potential contract.  Since the action space is large (potentially infinite) and has a well-defined real-valued structure, it is natural to consider an algorithm that uses \emph{discretization}. Our algorithm, \zooming, divides the action space into regions, and chooses among these regions, effectively treating each region as a single ``meta-arm.'' The discretization is defined \emph{adaptively}, so that the more promising areas of the action space are eventually discretized more finely than the less promising areas. While the general idea of adaptive discretization has appeared in prior work on MAB  \citep{LipschitzMAB-stoc08,xbandits-nips08,contextualMAB-colt11,ImplicitMAB-nips11}, our approach to adaptive discretization is new and problem-specific. The main difficulty, compared to this prior work, is that an algorithm is not given any information that links the observable numerical structure of contracts and the expected utilities thereof.

To analyze performance, we propose a concept called ``width dimension'' which measures how ``nice'' a particular problem instance is. We show that \zooming achieves regret sublinear in the time horizon for problem instances with small width dimension.  In particular, if the width dimension is $d$, it achieves regret $O(\log T \cdot T^{(d+1)/(d+2)})$ after $T$ rounds.  For problem instances with large width dimension, \zooming matches the performance of the naive algorithm which uniformly discretizes the space and runs a standard bandit algorithm. We illustrate our general results via some corollaries and special cases, including the \hlexample described above.
%an improvement over the prior work on \taskPricing.
We support the theoretical results with simulations.

Further, we consider a special case of our setting where each worker only chooses whether to accept or reject a given task. This special case corresponds to a dynamic pricing problem previously studied in the literature. Our results significantly improve over the prior work on this problem.

Our contributions can be summarized as follows. We define a broad, practically important setting in  crowdsourcing markets; identify novel problem-specific structure, for both the algorithm and the regret bounds; distill ideas from prior work to work with these structures; argue that our approach is productive by deriving corollaries and comparing to prior work; and identify and analyze specific examples where our theory applies. The main conceptual contributions are the model itself and the adaptive discretization approach mentioned above. Finally, this paper prompts further research on \problem along several directions that we outline in the conclusion.

\xhdr{Related work.}
Our work builds on three areas of research.  First, our model can be viewed as a multi-round version of the classical \emph{principal-agent} model from contract theory~\citep{LM02}. A single round of our model corresponds to the basic principal-agent setting, with \emph{adverse selection} (unknown worker's type) and \emph{moral hazard} (unobservable worker's decisions). Unlike much of the prior work in contract theory, the prior over worker types is not known to the principal, but may be learned over time. Accordingly, our techniques are very different from those employed in contract theory.

\OMIT{ %%%%%%
In the most basic principal-agent setting, a single principal interacts with a single agent.  The principal first selects a \emph{contract} which specifies a mapping from \emph{outcomes} to payments he will make to the agent.  The agent then selects among a set of \emph{actions}, each of which has some associated cost (e.g., in terms of effort) and results in a particular distribution over outcomes; it is assumed that the agent chooses her action to maximize her expected utility (payment from the principal minus cost of action) given the offered contract, and that choice is not directly observable by the requester. The principal has an associated value for each outcome, and his goal is to select a contract to maximize his own expected utility (value minus payment) given the action of the agent.   While the techniques we use are very different from those employed in contract theory, our utility model is heavily influenced by this line of work; see Appendix~\ref{sec:relatedcontract} for more.
} %%%%

Second, our methods build on those developed in the rich literature on MAB with continuous outcome spaces.  The closest line of work is that on \emph{Lipschitz MAB}~\citep{LipschitzMAB-stoc08}, in which the algorithm is given a distance function on the arms, and the expected rewards of the arms are assumed to satisfy Lipschitz-continuity (or a relaxation thereof) with respect to this distance function, ~\citep{Agrawal-bandits-95,Bobby-nips04,AuerOS/07,LipschitzMAB-stoc08,xbandits-nips08,contextualMAB-colt11}. Most related to our techniques is the idea of adaptive discretization \citep{LipschitzMAB-stoc08,xbandits-nips08,contextualMAB-colt11}, and in particular, the \emph{zooming algorithm} \citep{LipschitzMAB-stoc08,contextualMAB-colt11}.  However, the zooming algorithm cannot be applied directly in our setting because the required numerical similarity information is not immediately available. This problem also arises in web search and advertising, where it is natural to assume that an algorithm can only observe a tree-shaped taxonomy on arms~\citep{Kocsis-ecml06,Munos-uai07,yahoo-bandits07} which can be used to explicitly reconstruct relevant parts of the underlying metric space~\citep{ImplicitMAB-nips11,AdamBull-13}.  We take a different approach, using a notion of ``virtual width'' to estimate similarity information.  Explicit comparisons between our results and prior MAB work are made throughout the paper.

Finally, our work follows several other theoretical papers on pricing in crowdsourcing markets \citep{KleinbergL03,DynProcurement-ec12,SM13,Krause-www13,BwK-focs13}. In particular, \citet{DynProcurement-ec12} and \citet{Krause-www13} study a version of our setting with simple, single-price contracts (independent of the output), where the focus is on dealing with a global budget constraint.

\OMIT{\asedit{Empirical work \citep{MW09,YCS13,Harris11,PBPs-www15} finds that workers do respond to a change in financial incentives, but not quite as predicted by the principal-agent model.%
\footnote{\citet{PBPs-www15} is follow-up work w.r.t. the conference publication of this paper.}
\citet{PBPs-www15} suggest a generalized worker model which is consistent with the experimental evidence. Reassuringly, our results hold under that model as well.}}
 
A more thorough literature review (including a discussion of some related empirical work) can be found in Section~\ref{sec:related-work}.

\section{Our setting: the \problem problem}
\label{sec:setting}
In this section, we formally define the problem that we set out to solve. We start by describing a \emph{static model}, which captures what happens in a single round of interaction between a requester and a worker.  As described above, this is a version of the standard \emph{principal-agent} model~\citep{LM02}.
%
% which is reviewed in Appendix~\ref{sec:relatedcontract}.
%
We then define our \emph{dynamic model}, an extension of the static model to multiple rounds, with a new worker arriving each round. We then detail the objective of our pricing algorithm and the simplifying assumptions that we make throughout the paper. Finally, we compare our setting to the classic multi-armed bandit problem.

\xhdr{Static model.} We begin with a description of what occurs during each interaction between the requester and a single worker.
%
%(or in standard contract-theoretic terms, the principal and a single agent).
%
The requester first posts a task which may be completed by the worker, and a \emph{contract} specifying how the worker will be paid if she completes the task. If the task is completed, the requester pays the worker as specified in the contract, and the requester derives value from the completed task; for normalization, we assume that the value derived is in $[0,1]$. The requester's utility from a given task is this value minus the payment to the worker.

When the worker observes the contract and decides whether or not to complete the task, she also chooses a level of effort to exert, which in turn determines her cost (in terms of time, energy, or missed opportunities) and a distribution over the quality of her work. To model quality, we assume that there is a (small) finite set of possible \emph{outcomes} that result from the worker completing the task (or choosing not to complete it), and that the realized outcome determines the value that the requester derives from the task. The realized outcome is observed by the requester, and the contract that the requester offers is a mapping from outcomes to payments for the worker.

We emphasize two crucial (and related) features of the principal-agent model: that the mapping from effort level to outcomes can be randomized, and that the effort level is not directly observed by the requester. This is in line with a standard observation in crowdsourcing that even honest, high-effort workers occasionally make errors.

The worker's utility from a given task is the payment from the requester minus the cost corresponding to her chosen effort level. Given the contract she is offered, the worker chooses her effort level strategically so as to maximize her expected utility. Crucially, the chosen effort level is not directly observable by the requester.

The worker's choice \emph{not} to perform a task is modeled as a separate effort level of zero cost (called the \emph{null} effort level) and a separate outcome of zero value and zero payment (called the \emph{null} outcome) such that the null effort level deterministically leads to the null outcome, and it is the only effort level that can lead to this outcome.
% The payment for the null outcome must be $0$.

The mapping from outcomes to the requester's value is called the requester's \emph{value function}.  The mapping from effort levels to costs is called the \emph{cost function}, and the mapping from effort levels to distributions over outcomes is called the \emph{production function}. For the purposes of this paper, a worker is completely specified by these two functions; we say that the cost function and the production function comprise the worker's \emph{type}.  Unlike some traditional versions of the principal-agent problem, in our setting a worker's type is not  observable by the requester, nor is any prior given.

\xhdr{Dynamic model.}
The dynamic model we consider in this paper is a natural extension of the static model to multiple rounds and multiple workers.  We are still concerned with just a single requester. In each round, a new worker arrives.  We assume a \emph{stochastic environment} in which the worker's type  in each round is an i.i.d. sample from some fixed  and unknown distribution over types, called the \emph{supply distribution}.  The requester posts a new task and a contract for this task. All tasks are of the same type, in the sense that the set of possible effort levels and the set of possible outcomes are the same for all tasks. The worker strategically chooses her effort level so as to maximize her expected utility from this task. Based on the chosen effort level and the worker's production function, an outcome is realized.  The requester observes this outcome (but not the worker's effort level) and pays the worker the amount specified by the contract. The type of the arriving worker is never revealed to the requester. The requester can adjust the contract from one round to another, and his total utility is the sum of his utility over all rounds. For simplicity, we assume that the number of rounds is known in advance, though this assumption can be relaxed using standard tricks.

\xhdr{The \problem problem.}  Throughout this paper, we take the point of view of the requester interacting with workers in the dynamic model.  The algorithms we examine dynamically choose contracts to offer on each round with the goal of maximizing the requester's expected utility.  A problem instance consists of several quantities, some of which are known to the algorithm, and some of which are not. The known quantities are the number of outcomes, the requester's value function, and the time horizon $T$ (i.e., the number of rounds). The latent  quantities are the number of effort levels, the set of worker types, and the supply distribution. The algorithm adjusts the contract from round to round and observes the realized outcomes but receives no other feedback.

We focus on contracts that are \emph{bounded} (offer payments in $[0,1]$), and \emph{monotone} (assign equal or higher payments for outcomes with higher value for the requester). Let $X$ be the set of all bounded, monotone contracts. We compare a given algorithm against a given subset of ``candidate contracts'' $\Xcand \subset X$. Letting $\OPT(\Xcand)$ be the optimal utility over all contracts in $\Xcand$, the goal is to minimize the algorithm's \emph{regret}
    $R(T|\Xcand)$,
defined as $T\times\OPT(\Xcand)$ minus the algorithm's expected utility.

The subset $\Xcand$ may be finite or infinite, possibly $\Xcand = X$. The most natural example of a finite $\Xcand$ is the set of all bounded, monotone contracts with payments that are integer multiples of some $\minSize>0$; we call it the \emph{uniform mesh} with granularity $\minSize$, and denote it $\Xcand(\minSize)$.

\OMIT{\JWVedit{Another natural choice would be the set of all contracts in $\Xcand(\minSize)$ that offer at least some minimal ``fair'' payment to workers who complete the task (for example, minimum wage times the expected task completion time).\footnote{To ensure that all workers receive this minimum payment, it would additionally be necessary to redefine $X$ to be the set of all bounded, monotone contracts satisfying this property. The corresponding modifications to our algorithm and analysis are straightforward.}}}

\xhdr{Notation.}
Let $v(\cdot)$ be the value function of the requester, with $v(\pi)$ denoting the value of outcome $\pi$. Let $\mO$ be the set of all outcomes and let $m$ be the number of non-null outcomes. We will index the outcomes as
    $\mO = \{0,1,2 \LDOTS m\}$
in the order of increasing value (ties broken arbitrarily), with a convention that $0$ is the null outcome.

Let $c_i(\cdot)$ and $f_i(\cdot)$ be the cost function and production function for type $i$.  Then the cost of choosing effort level $e$ is $c_i(e)$, and the probability of obtaining outcome $\pi$ having chosen effort $e$ is $f_i(\pi|e)$. Let
    $F_i(\pi|e)=\sum_{\pi'\geq \pi}f_i(\pi'|e)$.

Recall that a contract $x$ is a function from outcomes to (non-negative) payments. If contract $x$ is offered to a worker sampled i.i.d. from the supply distribution, $V(x)$ is the expected value to the requester, $P(x)\geq 0$ is the expected payment, and
    $U(x) = V(x)-P(x)$
is the expected utility of the requester. Let
    $\OPT(\Xcand) = \sup_{x\in \Xcand} U(x)$.

\ignore{
Suppose a worker of type $i$ is offered a contract $x$ and chooses effort level $e$. Then
the expected payment is
    $P_i(x|e) = \sum_{\pi\in \mO} x(\pi)\, f_i(\pi|e)$,
and the algorithm's expected value is
    $V_i(x|e) = \sum_{\pi\in \mO} v(\pi)\, f_i(\pi|e)$.
Accordingly,  the worker's expected utility is
    $U_i(x|e) = P_i(x|e) - c_i(e)$,
and the algorithm's expected utility is
    $U^*_i(x|e) = V_i(x|e) - P_i(x|e)$. \jenn{The $*$ notation makes me think ``optimal,'' not ``algorithm.'' Do we need this notation?}

\ascomment{We are not using this notation except maybe in some examples, so I suggest to remove this from here and define where and if we actually use it.}
There are $n\leq \infty$ types of workers. The fraction of type $i$ in the worker population is denoted $\lambda_i$. ($\lambda_i$ is the probability density
function if there are infinitely many types.)
}

\xhdr{Assumption: First-order stochastic dominance (FOSD).}
Given two effort levels $e$ and $e'$, we say that $e$ has FOSD over $e'$ for type $i$ if $F_i(\pi|e) \geq F_i(\pi|e')$ for all outcomes $\pi$, with a strict inequality for at least one outcome.%
\footnote{This mimics the standard notion of FOSD between two distributions over a linearly ordered set.}
We say that type $i$ satisfies the FOSD assumption if for any two distinct effort levels, one effort level has FOSD over  the other for type $i$. We assume that all types satisfy this assumption.

\xhdr{Assumption: Consistent tie-breaking.} If multiple effort levels maximize the expected utility of a given worker for a contract $x$, we assume the tie is broken consistently in the sense that this worker chooses the same effort level for any contract that leads to this particular tie. This assumption is minor; it can be avoided (with minor technical complications) by adding random perturbations to the contracts. This assumption is implicit throughout the paper.

\OMIT{\xhdr{Our benchmark.}
We compare a given algorithm with the \emph{omniscient benchmark}, an algorithm which optimizes its actions over time given the full knowledge of the latent information in the problem instance. This is a standard benchmark for many machine learning problems. In our setting this benchmark reduces to the best fixed contract from a set, that is, the contract that maximizes the requester's expected utility in any given round.}

\subsection{Discussion}
\label{subsec:discussion}

\xhdr{Number of outcomes.}  Our results assume a small number of outcomes. This regime is important in practice, as the quality of submitted work is typically difficult to evaluate in a very fine granularity. Even with $m=2$ non-null outcomes, our setting has not been studied before.
The special case $m=1$ is equivalent to the dynamic pricing problem from \citet{Bobby-focs03}; we obtain improved results for it, too.

\xhdr{The benchmark.}
Our benchmark $\OPT(\cdot)$ only considers contracts that are bounded and monotone. In practice, restricting to such contracts may be appealing to all human parties involved. However, this restriction is not without loss of generality: there are problem instances in which monotone contracts are not optimal; see Appendix~\ref{sec:example-monotone} for an example. Further, it is not clear whether bounded monotone contracts are optimal among monotone contracts.

Our benchmark $\OPT(\Xcand)$ is relative to a given set $\Xcand$, which is typically a finite discretization of the contract space. There are two reasons for this. First, crowdsourcing platforms may require the payments to be multiples of some minimum unit (e.g., one cent), in which case it is natural to restrict our attention to contracts satisfying the same constraint. Second, achieving guarantees relative to $\OPT(X)$ for the full generality of our problem appears beyond the reach of our techniques. As in many other machine learning scenarios, it is useful to consider a restricted ``benchmark set'' -- set of alternatives to compare to.%
\footnote{A particularly relevant analogy is contextual bandits with policy sets, e.g., \citet{policy_elim}.}
In such settings, it is considered important to handle \emph{arbitrary} benchmark sets, which is what we do.

One known approach to obtain guarantees relative to $\OPT(X)$ is to start with some finite $\Xcand\subset X$, design an algorithm with guarantees relative to
    $\OPT(\Xcand)$,
and then, as a separate result, bound the discretization error
    $\OPT(X)-\OPT(\Xcand)$.
Then the choice of $\Xcand$ drives the tradeoff between the discretization error and regret $R(T|\Xcand)$, and one can choose $\Xcand$ to optimize this tradeoff. However, while one can upper-bound the discretization error in some (very) simple special cases (see Section~\ref{sec:hlexample}), it is unclear whether this can be extended to the full generality of \problem.

\xhdr{Alternative worker models.}
One of the crucial tenets in our model is that the workers maximize their expected utility. This ``rationality assumption" is very standard in Economics, and is often used to make the problem amenable to rigorous analysis. However, there is a considerable literature suggesting that in practice workers may deviate from this ``rational" behavior. Thus, it is worth pointing out that our results do not rely heavily on the rationality assumption. The FOSD assumption (which is also fairly standard) can be circumvented, too. In fact, all our assumptions regarding worker behavior serve  \emph{only} to enable us to prove Lemma~\ref{lm:virtual_width}, and more specifically to guarantee that the collective worker behavior satisfies the following natural property (which is used in the proof of  Lemma~\ref{lm:virtual_width}): if the requester increases the ``increment payment'' (as described in the next section) for a particular outcome, the probability of obtaining an outcome at least that good also increases.

\xhdr{Minimum wage.} For ethical or legal reasons one may want to enforce some form of minimum wage. This can be expressed within our model as a \emph{minimal payment} $\theta$ for a completed task, i.e., for any non-null outcome. Our algorithm can be easily modified to accommodate this constraint. Essentially, it suffices to restrict the action space to contracts that pay at least $\theta$ for a completed task. Formally, the ``increment space" defined in Section~\ref{sec:algorithm} should be
    $[\theta,1]\times [0,1]^{m-1}$
rather than $[0,1]^m$, and the ``quadrants" of each ``cell" are defined by splitting the cell in half in each dimension. All our results easily carry over to this version (restricting $\Xcand$ to contracts that pay at least $\theta$ for a completed task). We omit further discussion of this issue for the sake of simplicity.

\xhdr{Comparison to multi-armed bandits (MAB).}
\Problem can be modeled as special case of the MAB problem with some additional, problem-specific structure.  The basic MAB problem is defined as follows. An algorithm repeatedly chooses actions from a fixed \emph{action space} and collects rewards for the chosen actions; the available actions are traditionally called \emph{arms}. More specifically, time is partitioned into rounds, so that in each round the algorithm selects an arm and receives a reward for the chosen arm. No other information, such as the reward the algorithm would have received for choosing an alternative arm, is revealed. In an MAB problem with \emph{stochastic rewards}, the reward of each arm in a given round is an i.i.d. sample from some distribution which depends on the arm but not on the round. A standard measure of algorithm's performance is regret with respect to the best fixed arm, defined as the difference in expected total reward between a benchmark (usually the best fixed arm) and the algorithm.

Thus, \problem can be naturally modeled as an MAB problem with stochastic rewards, in which arms correspond to monotone contracts. The prior work on MAB with large / infinite action spaces often assumes known upper bounds on similarity between arms. More precisely, this prior work would assume that an algorithm is given a metric $\mD$ on contracts such that expected rewards are Lipschitz-continuous with respect to $\mD$, i.e., we have upper bounds
    $|U(x)-U(y)| \leq \mD(x,y)$
for any two contracts $x,y$.%
\footnote{Such upper bound is informative if and only if $\mD(x,y)<1$.}
However, in our setting such upper bounds are absent. On the other hand, our problem has some supplementary structure compared to the standard MAB setting. In particular, the algorithm's reward decomposes into value and payment, both of which are determined by the outcome, which in turn is probabilistically determined by the worker's strategic choice of the effort level. Effectively, this supplementary structure provides some ``soft" information on similarity between contracts, in the sense that numerically similar contracts are usually (but not always) similar to one another.

\OMIT{ %%%%
As discussed in the introduction, there is a considerable amount of work on MAB problems with large, structured action spaces.  However, this prior work typically assumes that the structure provides explicit upper bounds on the similarity between the contracts
(in terms of their expected utilities), whereas in our setting such explicit upper bounds are absent.  See Appendix~\ref{sec:relatedbandits} for more detail.
} %%%%%%

\section{Our algorithm: \protect\zooming}
\label{sec:algorithm}
In this section, we specify our algorithm. We call it \zooming because it ``zooms in'' on more promising areas of the action space, and does so without knowing a precise measure of the similarity between contracts.  This zooming can be viewed as a dynamic form of discretization.  Before stating the algorithm itself, we discuss the discretization of the action space in more detail, laying the groundwork for our approach.

\subsection{Discretization of the action space}
\label{sec:discretization}

In each round, the \zooming algorithm partitions the action space into several regions and chooses among these regions, effectively treating each region as a ``meta-arm.'' In this section, we discuss which subsets of the action space are used as regions, and introduce some useful notions and properties of such subsets.

%\ascomment{Deleted the para on candidate contracts because we say it already Section 2.}
\OMIT{ %%%%
\xhdr{Candidate contracts.}
Our algorithm is parameterized by a subset $\Xcand$ of bounded, monotone contracts. (The most natural example is the uniform mesh $\Xcand(\minSize)$: the set of all bounded, monotone contracts with payments that are integer multiples of some $\minSize>0$.) Contracts in $\Xcand$ are called \emph{candidate contracts}. The algorithm's goal is to compete against $\OPT(\Xcand)$.
} %%%%%%

\xhdr{Increment space and cells.}
To describe our approach to discretization, it is useful to think of contracts in terms of \emph{increment payments}. Specifically, we represent each monotone contract $x:\mO\to [0,\infty)$ as a vector
    $\vec{x}\in [0,\infty)^m$,
where $m$ is the number of non-null outcomes and $\vec{x}_\pi = x(\pi) -x(\pi-1)\geq 0$ for each non-null outcome $\pi$.
(Recall that by convention $0$ is the null outcome and $x(0) = 0$.)
We call this vector the \emph{increment representation} of contract $x$, and denote it $\incr(x)$. Note that if $x$ is bounded, then $\incr(x)\in [0,1]^m$. Conversely, call a contract \emph{weakly bounded} if it is monotone and its increment representation lies in $[0,1]^m$. Such a contract is not necessarily bounded.

We discretize the space of all weakly bounded contracts, viewed as a multi-dimensional unit cube. More precisely, we define the \emph{increment space} as $[0,1]^m$ with a convention that every vector represents the corresponding weakly bounded contract. Each region in the discretization is a closed, axis-aligned $m$-dimensional cube in the increment space; henceforth, such cubes are called \emph{cells}. A cell is called \emph{relevant} if it contains at least one candidate contract. A relevant cell is called \emph{atomic} if it contains exactly one candidate contract, and \emph{composite} otherwise.

In each composite cell $C$, the algorithm will only use two contracts:
the \emph{maximal corner}, denoted $x^+(C)$, in which all increment
payments are maximal, and the \emph{minimal corner}, denoted $x^-(C)$,
in which all increment payments are minimal. These two contracts are called the
\emph{anchors} of $C$.
In each atomic cell $C$, the algorithm will only use one contract: the unique candidate contract, also called the \emph{anchor} of $C$.

\xhdr{Virtual width.}
To take advantage of the problem structure, it is essential to estimate how similar the contracts within a given composite cell $C$ are. Ideally, we would like to know the maximal difference in expected utility:
    $$ \width(C) = \textstyle \sup_{x,y\in C}\; \abs{U(x) -U(y)} .$$
We estimate the width using a proxy, called \emph{virtual width}, which is expressed in terms of the anchors:
\begin{align}\label{eq:virt-width}
\vw(C) = \left( V(x^+(C)) - P(x^-(C)) \right) - \left( V(x^-(C)) -P(x^+(C)) \right).
\end{align}
%where $x^+=x^+(C)$ and $x^-=x^-(C)$ are the two anchors.

This definition is one crucial place where the problem structure is used. (Note that it is \emph{not} the difference in utility at the anchors.) It is useful due to the following lemma (proved in Section~\ref{sec:width-proof}).

\begin{lemma}\label{lm:virtual_width}
If all types satisfy the FOSD assumption and consistent tie-breaking holds, then
    $\width(C)\leq \vw(C)$
for each composite cell $C$.
\end{lemma}

Recall that the proof of this lemma is the only place in the paper where we use our assumptions on worker behavior.  All further developments hold for any model of worker behavior which satisfies Lemma~\ref{lm:virtual_width}.

\subsection{Description of the algorithm}

With these ideas in place, we are now ready to describe our algorithm.  The high-level outline of \zooming is very simple. The algorithm maintains a set of \emph{active} cells which cover the increment space at all times. Initially, there is only a single active cell comprising the entire increment space. In each round $t$, the algorithm chooses one active cell $C_t$ using an upper confidence index and posts contract $x_t$ sampled uniformly at random among the anchors of this cell. After observing the feedback, the algorithm may choose to \emph{zoom in} on $C_t$, removing $C_t$ from the set of active cells and activating all relevant quadrants thereof, where the \emph{quadrants} of cell $C$ are defined as the $2^m$ sub-cells of half the size for which one of the corners is the center of $C$.
In the remainder of this section, we specify how the cell $C_t$ is
chosen (the \emph{selection rule}), and how the algorithm decides
whether to zoom in on $C_t$ (the \emph{zooming rule}).

Let us first introduce some notation. Consider cell $C$ that is active in
some round $t$. Let $U(C)$ be the expected utility from a single round
in which $C$ is chosen by the algorithm, i.e., the average expected
utility of the anchor(s) of $C$. Let $n_t(C)$ be the number of times
this cell has been chosen before round $t$. Consider all rounds in
which $C$ is chosen by the algorithm before round $t$. Let $U_t(C)$ be
the average utility over these rounds. For a composite cell $C$, let
$V^+_t(C)$ and $P^+_t(C)$ be the average value and average payment
over all rounds when anchor $x^+(C)$ is chosen. Similarly, let
$V^-_t(C)$ and $P^-_t(C)$ be the average value and average payment
over all rounds when anchor $x^-(C)$ is chosen. Accordingly, we can
estimate the virtual width of composite cell $C$ at time $t$ as
\begin{align}\label{eq:virt-width-estimate}
    W_t(C) =  \left( V^+_t(C) - P^-_t(C) \right) - \left( V^-_t(C) - P^+_t(C) \right).
\end{align}

To bound the deviations, we define the \emph{confidence radius} as
\begin{align}\label{eq:conf-rad}
    \rad_t(C) = \sqrt{\confC\,\log(T)/n_t(C)},
\end{align}
for some absolute constant $\confC$; in our analysis, $\confC\geq 16$
suffices. We
will show that with high probability all sample averages
defined above will stay within $\rad_t(C)$ of the respective
expectations. If this high probability event holds, the  width
estimate $W_t(C)$ will always be within $4\, \rad_t(C)$ of $\vw(C)$.

\xhdr{Selection rule.}
Now we are ready to complete the algorithm. The selection rule is as follows. In each round $t$, the algorithm chooses an active cell $C$ with maximal \emph{index} $I_t(\cdot)$. $I_t(C)$ is an upper confidence bound on the expected utility of \emph{any} candidate contract in $C$, defined as
\begin{align}\label{eq:index}
I_t(C) = \begin{cases}
   U_t(C) + \rad_t(C) & \text{if $C$ is an atomic cell},\\
   U_t(C) + W_t(C) + 5\,\rad_t(C) & \text{otherwise}.
\end{cases}
\end{align}

\xhdr{Zooming rule.}
We zoom in on a composite cell $C_t$ if
\[
	W_{t+1}(C_t)>5\,\rad_{t+1}(C_t),
\]
i.e., the uncertainty due to random sampling, expressed by the confidence radius, becomes sufficiently
small compared to the uncertainty due to discretization, expressed by the virtual width.  We never zoom in on atomic cells.
The pseudocode is summarized in Algorithm~\ref{alg:zooming}.

\newcommand{\TAB}{\hspace{8mm}}
\newcommand{\MyComment}[1]{$\backslash\backslash$ {#1}}
\begin{algorithm}[h]
\caption{\zooming}
\label{alg:zooming}
\begin{algorithmic}
\STATE {\bf Inputs:} subset $\Xcand\subset X$ of \emph{candidate contracts}.
%    constant $\confC$ in \refeq{eq:conf-rad}.
\STATE {\bf Data structure:} Collection $\mA$ of cells.
    Initially, $\mA = \{\, [0,1]^m \,\}$.
\STATE {\bf For} each round $t=1$ to $T$
\STATE\TAB Let $C_t = \argmax_{C\in \mA}\; I_t(C)$,
    where $I_t(\cdot)$ is defined as in \refeq{eq:index}.
\STATE\TAB Sample contract $x_t$ u.a.r. among the anchors of $C_t$.
%\STATE\TAB\TAB
    \MyComment{Anchors are defined in Section~\ref{sec:discretization}.}
\STATE\TAB Post contract $x_t$ and observe feedback.
\STATE\TAB {\bf If} $|C\cap \Xcand|> 1$ {\bf and}
    $5\, \rad_{t+1}(C_t) < W_{t+1}(C_t)$
    {\bf then}
\STATE\TAB\TAB
    $\mA \leftarrow
        \mA \cup \{ \text{all relevant quadrants of $C_t$}\} \setminus\{ C_t\}$.
%\STATE\TAB\TAB
    \MyComment{$C$ is \emph{relevant} if $|C\cap \Xcand|\geq 1$.} \\
\end{algorithmic}
\end{algorithm}

\xhdr{Integer payments.}
In practice it may be necessary to only allow contracts in which all payments are integer multiples of some amount $\minSize$, e.g., whole cents. (In this case we can assume that candidate contracts have this property, too.) Then we can redefine the two anchors of each composite cell: the maximal (resp., minimal) anchor is the nearest allowed contract to the maximal (resp., minimal) corner. Width can be redefined as a $\sup$ over all allowed contracts in a given cell. With these modifications, the analysis goes through without significant changes. We omit further discussion of this issue.

\OMIT{ %%%%%%
\asedit{The anchors may or may not be in the set $\Xcand$. In some settings, e.g. if $\Xcand$ is a uniform mesh, it may be preferable for an algorithm to only choose candidate contracts. Then one can try to redefine the anchors of each composite cell $C$ so that they are candidate contracts. Instead of the maximal corner $x^+(C)$ we can use the \emph{maximal anchor} $y^+(C)$, defined as follows. Pick some parameter $\delta>0$. If there exists a candidate contract $y$ in which all increment payments are at least as large as in $x^+(C)$ and $\|y-x^+(C)\|_1 \leq \delta$, then $y^+(C)$ is one such contract $y$ that minimizes $\|y-x^+(C)\|_1$. Else, $y^+(C) = x^+(C)$. The minimal anchor $y^-(C)$ can be defined similarly, and used instead of the minimal corner $x^-(C)$. With these modifications, the analysis goes through without significant changes. We omit further discussion of this issue.}
} %%%%%%

\subsection{Proof of Lemma~\ref{lm:virtual_width} (virtual width)}
\label{sec:width-proof}

%\begin{lemma}[Lemma~\ref{lm:virtual_width}, restated]
%If all types satisfy the FOSD assumption and consistent tie-breaking holds, then
%    $\width(C)\leq \vw(C)$
%for each cell $C$.
%\end{lemma}

For two vectors $\vec{x},\vec{x}'\in \Re^m$, write $\vec{x}' \succeq \vec{x}$ if $\vec{x}'$ pointwise dominates $\vec{x}$, i.e., if $\vec{x}'_j\geq \vec{x}_j$ for all $j$. For two monotone contracts $x,x'$, write $x'\succeq x$ if $\incr(x')\succeq \incr(x)$.

 \begin{claim}\label{cl:fosd}
Consider a worker whose type satisfies the FOSD assumption and two weakly bounded contracts $x,x'$ such that $x'\succeq x$. Let $e$ (resp., $e'$) be the effort levels exerted by this worker when he is offered contract $x$ (resp., $x'$). Then $e$ does not have FOSD over $e'$.
\end{claim}

\begin{proof}
For the sake of contradiction, assume that $e$ has FOSD over $e'$. Note that $e\neq e'$.

Let $i$ be the worker's type. Recall that $F_i(\pi|e)$ denotes the probability of generating an outcome $\pi'\geq \pi$ given the effort level $e$.  Define
    $\vec{F} = \left(\; F_i(1|e) \LDOTS F_i(m|e) \;\right)$,
and define $\vec{F}'$ similarly for $e'$.

Let  $\vec{x}$ and $\vec{x}'$ be the increment representations for $x$ and $x'$. Given contract $x$, the worker's expected utility for effort level $e$ is
$U_i(x|e) =  \vec{x} \cdot\vec{F} - c_i(e)$. Since $e$ is the optimal effort level given this contract, we have $U_i(x|e) \geq U_i(x|e') $, and therefore
\begin{align*}
	\vec{x} \cdot \vec{F} - \vec{x} \cdot\vec{F}' \geq  c_i(e)-c_i(e').
\end{align*}
Similarly, since $e'$ is the optimal effort level given contract $\vec{x}'$, we have
\begin{align*}
	\vec{x}'\cdot\vec{F}' - \vec{x}' \cdot\vec{F} \geq  c_i(e')-c_i(e).
\end{align*}
Combining the above two inequalities, we obtain
\begin{align}\label{eq:pf:cl:fosd}
%	\vec{x} \cdot(\vec{F}-\vec{F}') \geq  \vec{x}'\cdot (\vec{F}-\vec{F}').
   (\vec{x}-\vec{x}')\cdot (\vec{F}-\vec{F}') \geq 0.
\end{align}

Note that if \refeq{eq:pf:cl:fosd} holds with equality then
    $U_i(x|e) = U_i(x|e')$ and $U_i(x'|e) = U_i(x'|e')$,
so the worker breaks the tie between $e$ and $e'$ in a different way for two different contracts. This contradicts the consistent tie-breaking assumption. However, \refeq{eq:pf:cl:fosd} cannot hold with a strict equality, either, because
    $\vec{x}' \succeq \vec{x}$
and (since $e$ has FOSD over $e'$) we have $\vec{F}\succeq \vec{F}'$ and
    $\vec{F}_\pi > \vec{F}'_\pi$
for some outcome $\pi>0$. Therefore we obtain a contradiction, completing the proof.
\end{proof}

The proof of Claim~\ref{cl:fosd}  is the only place in the paper where we directly use the consistent tie-breaking assumption. (But the rest of the paper relies on this claim.)

\begin{claim}
Assume all types satisfy the FOSD assumption. Consider weakly bounded contracts $x,x'$ such that $x'\succeq x$. Then $V(x')\geq V(x)$ and $P(x')\geq P(x)$.
\end{claim}
\begin{proof}
Consider some worker, let $i$ be his type. Let $e$ and $e'$ be the chosen effort levels for contracts $\vec{x}$ and $\vec{x}'$, respectively. By the FOSD assumption, either $e=e'$, or $e'$ has FOSD over $e$, or $e$ has FOSD over $e'$. Claim~\ref{cl:fosd} rules out the latter possibility.

Define vectors $\vec{F}$ and $\vec{F'}$ as in the proof of Claim~\ref{cl:fosd}. Note that
    $\vec{F}' \succeq \vec{F}$.

Then
    $P = \vec{x}\cdot\vec{F}$ and $P' = \vec{x}'\cdot\vec{F}'$
is the expected payment for contracts $x$ and $x'$, respectively.
Further, letting $\vec{v}$ denote the increment representation of the requester's value for each outcome,
    $V = \vec{v}\cdot\vec{F}$ and $V' = \vec{v}\cdot\vec{F}'$
is the expected requester's value for contracts $x$ and $x'$, respectively.
Since $\vec{x}'\succeq\vec{x}$ and $\vec{F}' \succeq \vec{F}$, it follows that $P'\geq P$ and $V'\geq V$. Since this holds for each worker, this also holds in expectation over workers.
\end{proof}

To finish the proof of Lemma~\ref{lm:virtual_width}, fix a contract $x\in C$ and observe that $V(x^+) \geq V(x) \geq V(x^-)$ and $P(x^+) \geq P(x) \geq P(x^-)$,
where $x^+=x^+(C)$ and $x^-=x^-(C)$ are the two anchors.

\section{Regret bounds and discussion}
\label{sec:regret-bounds}
%Before diving into the technical details of our analysis,

%We now formally state our results and discuss their significance.
We present the main regret bound for \zooming. Formulating this result requires some new, problem-specific structure. Stated in terms of this structure, the result is somewhat difficult to access. To explain its significance, we state several corollaries, and compare our results to prior work.

\xhdr{The main result.} We start with the main regret bound. Like the algorithm itself, this regret bound is parameterized by the set $\Xcand$ of candidate contracts; our goal is to bound the algorithm's regret with respect to candidate contracts.

%Let $X$ be the set of all weakly bounded monotone contracts.

Recall that
    $\OPT(\Xcand) = \sup_{x\in \Xcand} U(x)$
is the optimal expected utility over candidate contracts. The algorithm's regret with respect to candidate contracts is
    $R(T|\Xcand) = T\,\OPT(\Xcand)-U$,
where $T$ is the time horizon and $U$ is the expected cumulative utility of the algorithm.

Define the \emph{badness} $\Delta(x)$ of a contract $x\in X$ as the difference in expected utility between an optimal candidate contract and $x$:
    $\Delta(x) = \OPT(\Xcand) - U(x)$.
Let
	$X_\eps = \{x \in \Xcand:\, \Delta(x) \leq \eps\}$.
%comprise all candidate contracts with badness at most $\eps$.

We will only be interested in cells that can potentially be used by \zooming. Formally, we recursively define a collection of \emph{feasible} cells as follows: (i) the cell $[0,1]^m$ is feasible, (ii) for each feasible cell $C$, all relevant quadrants of $C$ are feasible. Note that the definition of a feasible cell implicitly depends on the set $\Xcand$ of candidate contracts.

Let $\mF_\eps$ denote the collection of all feasible, composite cells $C$ such that $\vw(C)\geq \eps$. For $Y\subset \Xcand$, let $\mF_\eps(Y)$ be the collection of all cells $C\in \mF_\eps$ that overlap with $Y$, and let
    $N_\eps(Y) = \abs{\mF_\eps(Y)}$; sometimes we will write $N_\eps(Y|\Xcand)$ in place of $N_\eps(Y)$
to emphasize the dependence on $\Xcand$.

Using the structure defined above, the main theorem is stated as follows. We prove this theorem in Section~\ref{sec:main-proof}.

\begin{theorem}\label{thm:regret}
Consider the \problem problem with all types satisfying the FOSD assumption and a constant number of outcomes. Consider \zooming, parameterized by some set $\Xcand$ of candidate contracts. Assume
    $T\geq\max(2^m+1,18)$.
There is an absolute constant $\SlackC>0$ such that for any $\delta>0$,
%the algorithm's regret satisfies
\begin{align}\label{eq:thm:regret}
R(T|\Xcand) \leq
    \delta T + O(\log T)
	       \sum_{\eps=2^{-j}\geq\delta:\; j\in \N}
                \frac{N_{\eps\,\SlackC}(X_{\eps}|\Xcand)}{\eps}.
\end{align}
\end{theorem}

\begin{remark}
As discussed in Section~\ref{subsec:discussion}, we target the practically important case of a small number of outcomes. The impact of larger $m$ is an exponential dependence on $m$ in the $O()$ notation, and, more importantly, increased number of candidate policies (typically exponential in $m$ for a given granularity).
\end{remark}

\begin{remark}
Our regret bounds do not depend on the number of worker types, in line with prior work on dynamic pricing. Essentially, this is because bandit approaches tend to depend only on expected reward of a given ``arm'' (and perhaps also on the variance), not the finer properties of the distribution.
\end{remark}

\refeq{eq:thm:regret} has a shape similar to several other regret bounds in the literature, as discussed below. To make this more apparent, we observe that regret bounds in ``bandits in metric spaces'' are often stated in terms of covering numbers. (For a fixed collection $\mF$ of subsets of a given ground set $X$, the \emph{covering number} of a subset $Y\subset X$ relative to $\mF$ is the smallest number of subsets in $\mF$ that is sufficient to cover $Y$.)
The numbers $N_\eps(Y | \Xcand)$ are, essentially, about covering $Y$ with feasible cells with virtual width close to $\eps$. We make this point more precise as follows. Let an \emph{$\eps$-minimal cell} be a cell in $\mF_\eps$ which does not contain any other cell in $\mF_\eps$. Let $\Nmin_\eps(Y)$ be the covering number of $Y$ relative to the collection of $\eps$-minimal cells, i.e., the smallest number of \eps-minimal cells sufficient to cover $Y$.
Then
\begin{align}\label{eq:Nmin}
N_\eps(Y)
        \leq \cel{\log\tfrac{1}{\minSize}}\; \Nmin_\eps(Y)
\;\text{for any $Y\subset \Xcand$ and $\eps\geq 0$},
\end{align}
where $\minSize$ is the smallest size of a feasible cell.%
\footnote{To prove \refeq{eq:Nmin}, observe that for each cell $C\in \mF_\eps(Y)$ there exists an $\eps$-minimal cell $C'\subset C$, and for each $\eps$-minimal cell $C'$ there exist at most $\cel{\log\tfrac{1}{\minSize}}$ cells $C\in \mF_\eps(Y)$ such that $C'\subset C$.}
Thus, \refeq{eq:thm:regret} can be easily restated using the covering numbers $\Nmin_\eps(\cdot)$ instead of $N_\eps(\cdot)$.

\OMIT{
\begin{claim}\label{cl:cov-nums-M}
$N_\eps(Y)
    \leq \cel{\log\tfrac{1}{\minSize}}\; \Nmin_\eps(Y)$
for any $Y\subset \Xcand$ and $\eps\geq 0$.
\end{claim}

\begin{proof}
For each cell $C\in \mF_\eps(Y)$, there exists an $\eps$-minimal cell $C'\subset C$. For each $\eps$-minimal cell $C'$, there exist at most $\cel{\log\tfrac{1}{\minSize}}$ cells $C\in \mF_\eps(Y)$ such that $C'\subset C$.
\end{proof}
} %%%%

\xhdr{Corollary: Polynomial regret.}
Literature on regret-minimization often states ``polynomial" regret bounds of the form
    $R(T) = \tilde{O}(T^{\gamma})$, $\gamma<1$.
%usually derived from a ``covering-number" regret bound such as
%\refeq{eq:thm:regret}.
While covering-number regret bounds are more precise and versatile, the exponent $\gamma$ in a polynomial regret bound expresses algorithms' performance in a particularly succinct and lucid way.

%Below we formulate a polynomial regret bound as a corollary of Theorem~\ref{thm:regret}.

For ``bandits in metric spaces'' the exponent $\gamma$ is typically determined by an appropriately defined notion of ``dimension", such as the covering dimension,%
\footnote{Given covering numbers $N_\eps(\cdot)$, the \emph{covering dimension} of $Y$ is the smallest $d\geq 0$ such that $N_\eps(Y) = O(\eps^{-d})$ for all $\eps>0$.}
which succinctly captures the difficulty of the problem instance.
Interestingly, the dependence of $\gamma$ on the dimension $d$ is typically of the same shape;
    $\gamma = (d+1)/(d+2)$,
for several different notions of ``dimension".  In line with this tradition, we define the \emph{width dimension}:
\begin{align} \label{eq:width-dim}
\WidthDim_\alpha = \inf \left\{ d\geq 0:\; N_{\eps\,\SlackC}(X_{\eps}|\Xcand)
    \leq \alpha\, \eps^{-d} \;
    \text{for all $\eps>0$} \right\}, \; \alpha>0.
\end{align}
Note that the width dimension depends on $\Xcand$ and the problem instance, and is parameterized by a constant $\alpha>0$. By optimizing the choice of $\delta$ in \refeq{eq:thm:regret}, we obtain the following corollary.

\begin{corollary}\label{cor:regret-dim}
Consider the the setting of Theorem~\ref{thm:regret}. For any $\alpha>0$, let $d=\WidthDim_\alpha$. Then
\begin{align}\label{eq:cor:regret-dim}
 R(T|\Xcand) \leq O(\alpha\,\log T)\; T^{(1+d)/(2+d)}.
\end{align}
\end{corollary}

The width dimension is similar to the ``zooming dimension" in \citet{LipschitzMAB-stoc08}
 and ``near-optimality dimension" in \citet{xbandits-nips08} in the work on ``bandits in metric spaces".
 %Appendix~\ref{sec:comparison-LischitzMAB}

\subsection{Comparison to prior work}

\xhdr{Non-adaptive discretization.}
One approach from prior work that is directly applicable to the
\problem problem is \emph{non-adaptive  discretization}.
This is an algorithm, call it \NonAdaptive, which runs an
off-the-shelf MAB algorithm, treating a set of candidate contracts $\Xcand$ as arms.%
\footnote{To simplify the proofs of the lower bounds, we assume that the candidate contracts are randomly permuted when given to the MAB algorithm.}
For concreteness, and following the prior work \citep{Bobby-focs03,Bobby-nips04,LipschitzMAB-stoc08}, we use a well-known algorithm \UCB~\citep{bandits-ucb1} as an off-the-shelf MAB algorithm.

To compare \zooming with \NonAdaptive, it is useful to derive several
``worst-case" corollaries of Theorem~\ref{thm:regret}, replacing
$N_{\eps}(X_{\eps})$ with various (loose) upper bounds.%
\footnote{We use the facts that $X_\eps\subset \Xcand$, $N_\eps(Y) \leq N_0(Y)$, and $\Nmin_0(Y)\leq |Y|$ for all subsets $Y\subset X$.}

\begin{corollary}\label{cor:regret-pessimistic}
In the setting of Theorem~\ref{thm:regret}, the regret of \zooming can be upper-bounded as follows:
\begin{OneLiners}
\item[(a)]
$R(T|\Xcand) \leq  \delta T +
    \sum_{\eps=2^{-j}\geq\delta:\; j\in \N}\,
        \tildeO(\abs{X_\eps}/\eps)$,
    for each $\delta\in(0,1)$.
\item[(b)]
$R(T|\Xcand) \leq  \tildeO( \sqrt{T\,\abs{\Xcand}} )$.
\end{OneLiners}
Here the $\tildeO()$ notation hides the logarithmic dependence on $T$ and $\delta$.
\end{corollary}

%\begin{align}
%R(T|\Xcand)
%    &\leq
%        \delta T + O(\log T)(\log \tfrac{1}{\delta})\;
%            \frac{N_{\delta\,\SlackC}(\Xcand)}{\delta},
%        \label{eq:cor:regret-pessimistic-1} \\
%R(T|\Xcand)
%    &\leq  \delta T + O(\log T)(\log \tfrac{1}{\minSize})\;
%    	    \sum_{\eps=2^{-j}\geq\delta:\; j\in \N}
%                \frac{\abs{X_\eps}}{\eps},
%         \label{eq:cor:regret-pessimistic-2} \\
%R(T|\Xcand)
%    &\leq O\left( \log^{3/2} (T)\; \log^{1/2} (\tfrac{1}{\minSize}) \right) \;\sqrt{T\,\abs{\Xcand}}.
%        \label{eq:cor:regret-pessimistic-3}
%\end{align}

\OMIT{ %%%%%
\begin{proof}
For \refeq{eq:cor:regret-pessimistic-1}, we replace $X_\eps$ with $\Xcand$ in \refeq{eq:thm:regret} and observe that the summands are maximized for $\eps=\delta$.
For \refeq{eq:cor:regret-pessimistic-2}, we apply \refeq{eq:Nmin} to \refeq{eq:thm:regret} and use the fact that $\Nmin_0(Y) \leq |Y|$ for all $Y$.
For \refeq{eq:cor:regret-pessimistic-3}, we apply \refeq{eq:Nmin} to \refeq{eq:cor:regret-pessimistic-1}  and optimize the $\delta$.
\end{proof}
} %%%

The best known regret bounds for \NonAdaptive coincide with those in
Corollary~\ref{cor:regret-pessimistic} up to poly-logarithmic
factors. However, the regret bounds in Theorem~\ref{thm:regret} may be significantly better than the ones in Corollary~\ref{cor:regret-pessimistic}.
We further discuss this in the next section, in the context of a specific example.

\OMIT{ %%%%%%
It is worth noting that in our setting it is not possible to choose $\Xcand$ so as to optimize the worst-case regret of \NonAdaptive. This is unlike the settings from prior work where non-adaptive discretization has been used (see Section~\ref{sec:related-work} for further discussion). \jenn{Why is it not possible?  I don't see this discussion in Section~\ref{sec:related-work}.}
} %%%%%%%

\OMIT{ %%%%%%
\newcommand{\Tub}{T^{\mathtt{UB}}} % upper bound on T

\xhdr{Regret on an arbitrary time interval.} The above results concern regret on the time interval $[0,T]$, for known time horizon $T$. However, it is easy to extend all these results to regret on an arbitrary time interval of duration $T$. More precisely, an arbitrary time interval $[T_0,T_0+T]$ such that $T_0+T\leq \Tub$ for some known upper bound $\Tub$. The only change is that the log factor changes from $\log(T)$ to $\log(\Tub)$, both in the algorithm's specification and in the regret bounds. \OMIT{A regret bound $R(T)$ that holds for any subinterval $[T_0,T_0+T]$ of the time interval $[0,\Tub]$ will be called \emph{uniform}.}
} %%%%%%

%\section{Comparison to prior work on ``bandits in metric spaces"}
%\label{sec:comparison-LischitzMAB}

\xhdr{Bandits in metric spaces.}
Consider a variant of \problem in which an algorithm is given a priori information on similarity between contracts: a function $\mD: \Xcand\times \Xcand \to[0,1]$ such that $|U(x)-U(y)| \leq \mD(x,y)$ for any two candidate contracts $x,y$. If an algorithm is given this function $\mD$ (call such algorithm \emph{$\mD$-aware}), the machinery from ``bandits in metric spaces" \cite{LipschitzMAB-stoc08,xbandits-nips08} can be used to perform adaptive discretization and obtain a significant advantage over $\NonAdaptive$. We argue that we obtain similar results with \zooming without knowing the $\mD$.

In practice, the similarity information $\mD$ would be coarse, probably aggregated according to some predefined hierarchy. To formalize this idea, the hierarchy can be represented as a collection $\mF$ of subsets of $\Xcand$, so that $\mD(x,y)$ is a function of the smallest subset in $\mF$ containing both $x$ and $y$. The hierarchy $\mF$ should be natural given the structure of the contract space. One such natural hierarchy is the collection of all feasible cells, which corresponds to splitting the cells in half in each dimension. Formally,
    $\mD(x,y) = f(C_{x,y})$ for some $f$ with $f(C_{x,y})\geq \width(C_{x,y})$,
where $C_{x,y}$ is the smallest feasible cell containing both $x$ and $y$.

Given this shape of $\mD$, let us state the regret bounds for $\mD$-aware algorithms in \citet{LipschitzMAB-stoc08} and \citet{xbandits-nips08}. To simplify the notation, we assume that the action space is restricted to $\Xcand$. The regret bounds have a similar ``shape'' as that in Theorem~\ref{thm:regret}:
\begin{align}\label{eq:regret-LipMAB}
R(T|\Xcand) \leq
    \delta T + O(\log T)
	       \sum_{\eps=2^{-j}\geq\delta:\; j\in \N}
                \frac{N^*_{\Omega(\eps)}(X_{\eps})}{\eps},
\end{align}
where the numbers $N^*_{\eps}(\cdot)$ have a similar high-level meaning as $N_\eps(\cdot)$, and nearly coincide with $\Nmin_\eps(\cdot)$ when $\mD(x,y) = \vw(C_{x,y})$. One can use \refeq{eq:regret-LipMAB} to derive a polynomial regret bound like \refeq{eq:cor:regret-dim}.

For a more precise comparison, we focus on the results in \citet{LipschitzMAB-stoc08} . (The regret bounds in \citet{xbandits-nips08} are very similar in spirit, but are stated in terms of a slightly different structure.) The ``covering-type" regret bound in \citet{LipschitzMAB-stoc08} focuses on balls of radius at most $\eps$ according to distance $\mD$, so that $N^*_{\eps}(Y)$ is the smallest number of such balls that is sufficient to cover $Y$. In the special case
    $\mD(x,y) = \vw(C_{x,y})$
balls of radius $\leq\eps$ are precisely feasible cells of virtual width $\leq\eps$. This is very similar (albeit not technically the same) as the \eps-minimal cells in the definition of
  $\Nmin_{\eps}(\cdot)$.

Further, the covering numbers $N^*_{\eps}(Y)$ determine the ``zooming dimension":
\begin{align} \label{eq:zooming-dim-defn}
\ZoomDim_\alpha = \inf \left\{ d\geq 0:\; N^*_{\eps/8}(X_{\eps})
    \leq \alpha\, \eps^{-d} \;
    \text{for all $\eps>0$} \right\}, \; \alpha>0.
\end{align}
This definition coincides with the covering dimension in the worst case, and can be much smaller for ``nice" problem instances in which $X_{\eps}$ is a significantly small subset of $\Xcand$. With this definition, one obtains a polynomial regret bound which is version of \refeq{eq:cor:regret-dim} with $d = \ZoomDim_\alpha$.

We conclude that \zooming essentially matches the regret bounds for $\mD$-aware algorithms, despite the fact that $\mD$-aware algorithms have access to much more information.

 % comparison vs MAB

%\section{Applications and special cases: the proofs}
%\label{app:applications}
\newcommand{\HighLowFormat}[1]{{\tt #1}\xspace}
\newcommand{\high}{\HighLowFormat{high}}
\newcommand{\High}{\HighLowFormat{High}}
\newcommand{\low}{\HighLowFormat{low}}
\newcommand{\Low}{\HighLowFormat{Low}}

\newcommand{\cH}{c_{\mathtt{h}}}
\newcommand{\vH}{v_{\mathtt{h}}}
\newcommand{\dH}{D_{\mathtt{h}}}

\section{A special case: the ``\hlexample''}
\label{sec:hlexample}

We apply the machinery in Section~\ref{sec:regret-bounds} on a special case, and we show that \zooming significantly outperforms \NonAdaptive.

The most basic special case is when there is just one non-null outcome. Essentially, each worker makes a strategic choice whether to accept or reject a given task (where ``reject" corresponds to the null effort level), and this choice is fully observable. This setting has been studied before \citep{Bobby-focs03,DynProcurement-ec12,Krause-www13,BwK-focs13}; we will call it \emph{\taskPricing}. Here the contract is completely specified by the price $p$ for the non-null outcome. The supply distribution is summarized by the function
    $S(p) =\Pr[ \text{accept}|p] $,
so that the corresponding expected utility is $U(p) = S(p)(v-p)$, where $v$ is the value for the non-null outcome. This special case is already quite rich, because $S(\cdot)$ can be an arbitrary non-decreasing function. By using adaptive discretization, we achieve significant improvement over prior work; see Section~\ref{sec:pricing} for further discussion.
%Appendix~\ref{sec:pricing}

\OMIT{and  its ``dual,'' \emph{dynamic pricing}, in which the principal is selling rather than buying~\cite{Bobby-focs03,BZ09,DynPricing-ec12,BwK-focs13}. }

\newcommand{\prH}{\theta_{\mathtt{h}}}

We consider a somewhat richer setting in which workers' strategic decisions are \emph{not} observable; this is a salient feature of our setting, called \emph{moral hazard} in the contract theory literature.
There are two non-null outcomes (\low and \high), and two non-null effort levels (\low and \high).
\Low outcome brings zero value to the requester, while $\high$ outcome brings value $v>0$.  \Low effort level inflicts zero cost on a worker and leads to low outcome with probability $1$. We assume that workers break ties between effort levels in a consistent way: \high better than \low better than {\tt null}. (Hence, as \low effort incurs zero cost, the only possible outcomes are \low and \high.)  We will call this the {\bf\em \hlexample}; it is perhaps the simplest example that features moral hazard.

In this example, the worker's type consists of a pair $(\cH,\prH)$, where $\cH\geq 0$ is the cost for \high effort and $\prH\in[0,1]$ is the probability of \high outcome given \high effort. Note that \taskPricing is equivalent to the special case $\prH=1$.

The following claim states a crucial property of the \hlexample.

\begin{claim}\label{cl:hlexample-general}
Consider the \hlexample with a fixed supply distribution. Then
the probability of obtaining \high outcome given contract $x$
    $\Pr[\text{\high outcome} \,|\, \text{contract $x$}]$
depends only on
    $p = x(\high)-x(\low)$;
denote this probability by $S(p)$. Moreover, $S(p)$ is non-decreasing in $p$. Therefore:
\begin{OneLiners}
\item expected utility is
    $U(x) = S(p) (v - p) - x(\low)$.

\item discretization error
    $\OPT(X)-\OPT(\Xcand(\minSize))$ is at most $3\minSize$, for any $\minSize>0$.
\end{OneLiners}
\end{claim}

Recall that $\Xcand(\minSize)$, the uniform mesh with granularity $\minSize>0$, consists of all bounded, monotone contracts with payments in $\minSize \N$.

For our purposes, the supply distribution is summarized via the function $S(\cdot)$. Denote $\tilde{U}(p) = S(p) (v - p)$. Note that $U(x)$ is maximized by setting $x(\low)=0$, in which case $U(x) = \tilde{U}(p)$. Thus, if an algorithm knows that it is given a \hlexample, it can set $x(\low)=0$, thereby reducing the dimensionality of the search space. Then the problem essentially reduces to \taskPricing with the same $S(\cdot)$.

However, in general an algorithm does not \emph{know} whether it is presented with the \hlexample (because the effort levels are not observable). So in what follows we will consider algorithms that do not restrict themselves to $x(\low)=0$.

\xhdr{``Nice" supply distribution.}
We focus on a supply distribution $D$ that is ``nice", in the sense that $S(\cdot)$ satisfies the following two properties:
\begin{OneLiners}
	\item $S(p)$ is Lipschitz-continuous: $|S(p) - S(p')| \leq L |p-p'|$  for some constant $L$.
	\item $\tilde{U}(p)$ is \emph{strongly concave}, in the sense that $\tilde{U}''(\cdot)$ exists and satisfies $\tilde{U}''(\cdot)\leq C<0$.
\end{OneLiners}
Here $L$ and $C$ are absolute constants.
We call such $D$ \emph{strongly Lipschitz-concave}.

The above properties are fairly natural. For example, they are satisfied if $\prH$ is the same for all worker types and the marginal distribution of $\cH$ is piecewise uniform such that the density is between $\tfrac{1}{\lambda}$ and $\lambda$, for some absolute constant $\lambda\geq 1$.

We show that for any choice $\Xcand\subset X$, \zooming has a small width dimension in this setting, and therefore small regret.

\begin{lemma}\label{lm:hlexample-nice}
Consider the \hlexample with a strongly Lipschitz-concave supply distribution. Then the width dimension is at most $\tfrac12$, for any given $\Xcand\subset X$.
Therefore, \zooming with this $\Xcand$ has regret $R(T|\Xcand) = O(\log T)\, T^{3/5}$.
\end{lemma}

We contrast this with the performance of \NonAdaptive, parameterized with the natural choice $\Xcand = \Xcand(\minSize)$. We focus on $R(T|X)$: regret w.r.t. the best contract in $X$. We show that \zooming achieves
    $R(T|X) = \tildeO(T^{3/5})$
for a wide range of $\Xcand$, whereas \NonAdaptive cannot do better than
    $R(T|X) = O(T^{3/4})$
for any $\Xcand = \Xcand(\minSize)$, $\minSize>0$.

\begin{lemma}\label{lm:hlexample-comparison}
Consider the setting of Lemma~\ref{lm:hlexample-nice}. Then:
\begin{OneLiners}
\item[(a)] \zooming with $\Xcand \supset \Xcand(T^{-2/5})$ has regret $R(T|X) = O(T^{3/5}\, \log T)$.
\item[(b)] \NonAdaptive with $\Xcand = \Xcand(\minSize)$ cannot achieve regret
    $R(T|X)< o(T^{3/4})$
    over all problem instances, for any $\minSize>0$.~
    \footnote{This lower bound holds even if \UCB in \NonAdaptive is replaced with any other MAB algorithm.}
\end{OneLiners}
\end{lemma}
%Moreover, the $T^{2/3}$ lower bound holds even if \NonAdaptive is tuned to the \hlexample so that
%    $\Xcand = \{ x\in \Xcand(\psi):\, x(\low)=0 \}$.

%\section{Proofs from Section~\ref{sec:hlexample} (the \hlexample)}
%\label{sec:hl_proof}
\subsection{Proofs}%

\begin{proof}[Proof of Claim~\ref{cl:hlexample-general}]
Consider a contract $x$ with $x(\low)=b$ and $x(\high)=b+p$, and a worker of type $(\cH,\prH)$. If the worker exerts \high effort, she pays cost $\cH$ and receives expected payment $\prH(p+b) + (1-\prH) b$, for a total expected payoff
    $p\prH + b - \cH$.
Her expected payoff for exerting \low effort is $b$. Therefore she will choose to exert \high effort if and only if  $p\prH+b-\cH \geq b$, i.e., if $\cH/\prH \leq p$, and choose to exert \low effort otherwise.
Therefore
\begin{align*}
    \Pr[\text{\high outcome} \,|\, \text{contract $x$}]
        = \Exp[(\cH,\prH)]{\prH\; \indicator{\cH/\prH \leq p}}.
\end{align*}
This is a function of $p$, call it $S(p)$. Moreover, this is a non-decreasing function simply because the expression inside the expectation is non-decreasing in $p$.

It trivially follows that
    $U(x) = S(p) (v - p) - x(\low)$.

We can upper-bound the discretization error using a standard approach from the work on dynamic pricing~\cite{KleinbergL03}. Fix discretization granularity $\minSize>0$. For any $\eps>0$, there exists a contract $x^*\in X$ such that $\OPT(X)-U(x^*)<\eps$. Round $x^*(\high)$ and $x^*(\low)$ up and down, respectively, to the nearest integer multiple of $\minSize$; let $x\in \Xcand(\minSize)$ be the resulting contract. Denoting
    $p = x(\high)-x(\low)$ and $p^* = x^*(\high)-x^*(\low)$,
we see that $p^* \leq p \leq p^*+2\minSize$. It follows that
    $$U(x) \geq U(x^*)-3\minSize \geq \OPT(X)-\eps-3\minSize.$$
Since this holds for any $\eps>0$, we conclude that
    $\OPT(X)-\OPT(\Xcand(\minSize)) \leq 3\minSize$.
\end{proof}

\begin{proof}[Proof of Lemma~\ref{lm:hlexample-nice}]
To calculate the width dimension, we need to count the number of feasible cells in the increment space which (i) has virtual width larger than or equal to $O(\epsilon)$ and (ii) overlaps with $X_\epsilon$, the set of contracts with badness smaller than $\epsilon$.

We first characterize $X_\epsilon$. We use $x_{p,b}$ to denote the contract with $x(\high)=p+b$ and $x(\low)=b$. The benefit of this representation is that, $p$ and $b$ would then be the two axis in the increment space. Let $x_{p^*, 0}$ be an optimal contract. Since $U(x_{p,b})$ is strongly concave in $p$, we know that for any $b$, there exists constants $C_1$ and $C_2$ such that for any $p\in[0,1]$, $C_1 (p^*-p)^2 \leq U(x_{p^*,b}) - U(x_{p,b}) \leq C_2 (p^*-p)^2$. Also we know that $U(x_{p^*,b}) = U(x_{p^*,0}) - b$. Therefore.
\[
	X_\epsilon = \{x_{p,b}: (p-p^*)^2 + b \leq O(\epsilon)\}
\]

We can also write it as
\[
	X_\epsilon = \{x_{p,b}: p^*-\prH(\sqrt{\epsilon}) \leq p \leq p^*+\prH(\sqrt{\epsilon}) \mbox{ and } b \leq O(\epsilon)\}
\]

Intuitively, $X_\epsilon$ contains contracts $\{x_{p,b}\}$ with $p$ not $O(\sqrt{\epsilon})$ away from $p^*$ and $b$ not $O(\epsilon)$ away from $b^*=0$.

Next we characterize the virtual width of a cell. We use $C_{p,b,d}$ to denote the cell with size $d$ and with anchors $\{x_{p,b}, x_{(p+d),(b+d)}\}$. We can derive the expected payment and value on the two anchors as:
\begin{OneLiners}
	\item $P^+(C_{p,b,d}) = (p+d) S(p+d) + b + d$
	\item $V^+(C_{p,b,d})= v S(p+d)$
	\item $P^-(C_{p,b,d}) = p S(p) + b$
	\item $V^-(C_{p,b,d}) = v S(p)$
\end{OneLiners}
\vspace{2mm}
By definition, we can get that (we use $d_F$ to represent $S(p+d)-S(p)$ for simplification)
\[ \vw(C_{p,b,d}) = (v+p) d_F + d S(p)  + d ~ d_F + d. \]

Now we can count the number of feasible cells with virtual width larger than $\prH(\epsilon)$ which overlaps with $X_\epsilon$. Note that  since the total number of feasible cells $C_{p,b,d}$ with large $d$ is small, we can treat the number of cells with large $d$ as a constant. Also, for any relevant cell $C_{p,b,d}$, we have $p\approx p^*$.
Therefore, we only care about feasible cells $C_{p,b,d}$ with small $d$ and when $p$ is close to $p^*$.
%and therefore we can treat $v+p$ and $S(p)$ as constants.

Since $S(p)$ is Lipschitz, we have $d_F = O(d)$. Therefore, for any relevant cell $C_{p,d}$,
\[ \vw(C_{p,b,d}) = O(d) \]

Given the above two arguments, we know that the number of cells with virtual width larger than $\epsilon$ which also overlaps with $X_\epsilon$ is $O(\epsilon/\epsilon) \times O(\sqrt{\epsilon}/\epsilon)=O(\epsilon^{-1/2})$.  Therefore the width dimension is $1/2$.
\end{proof}

\newcommand{\discrError}{\mathtt{Error}}

\begin{proof}[Proof Sketch of Lemma~\ref{lm:hlexample-comparison}(b)]
Consider a version of \NonAdaptive that runs an off-the-shelf MAB algorithm ALG on candidate contracts $\Xcand = \Xcand(\minSize)$. For ALG, the ``arms" are the candidate contracts; recall that the arms are randomly permuted before they are given to ALG.

Fix $\minSize>0$. It is easy to construct a problem instance with discretization error
    $\discrError \triangleq \OPT(X) - \OPT(\Xcand(\minSize)) \geq \Omega(\minSize)$.
Note that $\Xcand$ contains $N=\Omega(\minSize^{-2})$ suboptimal contracts that are suboptimal w.r.t. $\OPT(\Xcand)$.
(For example, all contracts $x$ with $x(\low)>0$ are suboptimal.)

Fix any problem instance $\mI$ of MAB with $N$ suboptimal arms. Using standard lower-bound arguments for MAB, one can show that if one runs ALG on a problem instance obtained by randomly permuting the arms in $\mI$, then the expected regret in $T$ rounds is at least $\Omega(\sqrt{NT})$.

Therefore, $R(T|\Xcand) \geq \Omega(\sqrt{NT})$. It follows that
$$ R(T|X)
    \geq \Omega(\sqrt{NT}) + \discrError \cdot T
    \geq \Omega( \sqrt{T}/\minSize + \minSize T)
    \geq \Omega( T^{3/4}). \qedhere
$$
\end{proof}
 % proofs for the H/L example

\section{Proof of the main regret bound (Theorem~\ref{thm:regret})}
\label{sec:main-proof}

We now prove the main result from Section~\ref{sec:regret-bounds}.
Our high-level approach is to define a \emph{clean execution} of an algorithm as an execution in which some high-probability events are satisfied, and derive bounds on regret conditional on the clean execution. The analysis of a clean execution does not involve any ``probabilistic'' arguments. This approach tends to simplify regret analysis.

We start by listing some simple invariants enforced by \zooming:

\begin{invariant}\label{cl:zooming-invariants}
In each round $t$ of each execution of \zooming:
\begin{OneLiners}
\item[(a)] All active cells are relevant,
\item[(b)] Each candidate contract is contained in some active cell,
\item[(c)] $W_t(C) \leq 5\,\rad_t(C)$ for each active composite cell $C$.
\end{OneLiners}
\end{invariant}

Note that the zooming rule is essential to ensure Invariant~\ref{cl:zooming-invariants}(c).

\subsection{Analysis of the randomness}

\begin{dfn}[Clean Execution]
An execution of \zooming is called \emph{clean} if for each round $t$ and each active cell $C$ it holds that
\begin{align}
	\abs{U(C) - U_t(C)} &\leq \rad_t(C), &       \label{eq:def-clean-U} \\
	\abs{\vw(C) - W_t(C)} &\leq 4\,\rad_t(C) & \text{(if $C$ is composite)}.  \label{eq:def-clean-W}
\end{align}
\end{dfn}

\begin{lemma}\label{lem:clean_prob}
Assume $\confC\geq 16$ and $T\geq\max(1+2^m,18)$. Then:
\begin{OneLiners}
\item[(a)]
$\Pr\left[\;
    \text{\refeq{eq:def-clean-U} holds}
    \quad\forall\, \text{rounds $t$, active cells $C$}
\;\right]
\geq 1-2\, T^{-2}.$

\item[(b)]
$\Pr\left[\;
    \text{\refeq{eq:def-clean-W} holds}
    \quad\forall\, \text{rounds $t$, active composite cells $C$}
\;\right]
\geq 1-16\, T^{-2}.$
\end{OneLiners}
Consequently, an execution of \zooming is clean with probability at least $1 - 1/T$.
\end{lemma}

Lemma~\ref{lem:clean_prob} follows from the standard concentration inequality known as ``Chernoff Bounds". However, one needs to be careful about conditioning and other details.

\begin{proof}[Proof of Lemma~\ref{lem:clean_prob}(a)]
Consider an execution of \zooming. Let $N$ be the total number of activated cells. Since at most $2^m$ cells can be activated in any one round,
    $N \leq 1 + 2^m T \leq T^2$.
Let $C_j$ be the $\min(j,N)$-th cell activated by the algorithm. (If multiple ``quadrants" are activated in the same round, order them according to some fixed ordering on the quadrants.)

Fix some feasible cell $C$ and $j\leq T^2$. We claim that
\begin{align}\label{eq:lem:clean_prob-C}
\Pr\left[\;
    |U(C)-U_t(C)| \leq \rad_t(C) \;\text{for all rounds $t$}\; \middle\vert \; C_j=C
\;\right]
\geq 1- 2\,T^{-4}.
\end{align}

Let $n(C) = n_{1+T}(C)$ be the total number of times cell $C$ is chosen by the algorithm. For each $s\in \N$: $1\leq s \leq n(C)$ let $U_s$ be the requester's utility in the round when $C$ is chosen for the $s$-th time. Further, let $\mD_C$ be the distribution of $U_1$, conditional on the event $n(S)\geq 1$. (That is, the per-round reward from choosing cell $C$.) Let
    $U'_1 \LDOTS U'_T$
be a family of mutually independent random variables, each with distribution $\mD_C$. Then for each $n \leq T$, conditional on the event $\{C_j=C\} \wedge \{n(C) = n\}$, the tuple $(U_1 \LDOTS U_n)$ has the same joint distribution as the tuple $(U'_1 \LDOTS U'_n)$. Consequently, applying Chernoff Bounds to the latter tuple, it follows that
\begin{align*}
&\textstyle \Pr\left[\;
    \big|U(C)-\frac{1}{n}\sum_{s=1}^n U_s\big| \leq \sqrt{\tfrac{1}{n}\,\confC\,\log(T)} \;\middle\vert\; \{C_j=C\} \wedge \{n(C) = n\}
\;\right] \\
&\qquad\qquad \geq 1-2\,T^{-2\confC} \geq 1-2\,T^{-5}.
\end{align*}
Taking the Union Bound over all $n\leq T$, and plugging in $\rad_t(C_j)$, $n_t(C_j)$, and $U_t(C_j)$, we obtain
\refeq{eq:lem:clean_prob-C}.

Now, let us keep $j$ fixed in \refeq{eq:lem:clean_prob-C}, and integrate over $C$. More precisely, let us multiply both sides of \refeq{eq:lem:clean_prob-C} by $\Pr[C_j=C]$ and sum over all feasible cells $C$. We obtain, for all $j\leq T^2$:
\begin{align}\label{eq:lem:clean_prob-Cj}
\Pr\left[\;
    \abs{U(C_j)-U_t(C_j)} \leq \rad_t(C_j) \;\text{for all rounds $t$}
\;\right]
\geq 1-2\,T^{-4}.
\end{align}
(Note that to obtain \refeq{eq:lem:clean_prob-Cj}, we do \emph{not} need to take the Union Bound over all feasible cells $C$.) To conclude, we take the Union Bound over all $j \leq 1+T^2$.
\end{proof}

\begin{proof}[Proof Sketch of Lemma~\ref{lem:clean_prob}(b)]
We show that
\begin{align}\label{eqn:clean-V}
\Pr\left[\;
    \abs{V^+(C)-V^+_t(C)} \leq \rad_t(C)
    \;\forall\, \text{rounds $t$, active composite cells $C$}
\;\right]
\geq 1-\tfrac{4}{T^2},
\end{align}
and similarly for $V^-()$, $P^+()$ and $P^-()$. Each of these four statements is proved similarly, using the technique from Lemma~\ref{lem:clean_prob}(a). In what follows, we sketch the proof for one of the four cases, namely for \refeq{eqn:clean-V}.

For a given composite cell $C$, we are only interested in rounds in which anchor $x^+(C)$ is selected by the algorithm. Letting $n^+_t(C)$ be the number of times this anchor is chosen up to time $t$, let us define the corresponding notion of ``confidence radius'':
$$ \rad^+_t(C) = \frac{1}{2}\, \sqrt{\frac{\confC\,\log T}{n^+_t(C)}}.
$$

With the technique from the proof of Lemma~\ref{lem:clean_prob}(a), we can establish the following high-probability event:
\begin{align}\label{eqn:clean-V-almost}
    \abs{V^+(C)-V^+_t(C)} \leq \rad^+_t(C).
\end{align}
More precisely, we can prove that
\begin{align*}
\Pr\left[\;
    \text{\refeq{eqn:clean-V-almost} holds}
    \quad\forall\, \text{rounds $t$, active composite cells $C$}
\;\right]
\geq 1-2\, T^{-2}.
\end{align*}

%\begin{align}\label{eqn:clean-V-almost}
%\Pr\left[\;
%    \abs{V^+(C)-V^+_t(C)} \leq \rad^+_t(C)
%    \quad\forall\, \text{rounds $t$, active composite cells $C$}
%\;\right]
%\geq 1-\tfrac{2}{T}.
%\end{align}

Further, we need to prove that w.h.p. the anchor $x^+(C)$ is played sufficiently often. Noting that
   $\E[n^+_t(C)] = \tfrac12\, n_t(C)$,
we establish an auxiliary high-probability event:%
\footnote{The constant $\tfrac14$ in \refeq{eq:clean-auxiliary} is there to enable a consistent choice of $n_0$ in the remainder of the proof.}
\begin{align}\label{eq:clean-auxiliary}
    n^+_t(C) \geq \tfrac12\, n_t(C) - \tfrac14\,\rad_t(C).
\end{align}
More precisely, we can use Chernoff Bounds to show that, if $\confC\geq16$,
\begin{align}\label{eq:clean-auxiliary-formal}
\Pr\left[\;
    \text{\refeq{eq:clean-auxiliary} holds}
    \quad\forall\, \text{rounds $t$, active composite cells $C$}
\;\right]
\geq 1-2\,T^{-2}.
\end{align}
Now, letting $n_0 = (\confC\,\log T)^{1/3}$, observe that
$$\begin{array}{lll}%\label{eq:clean-auxiliary}
    n_t(C) \geq n_0
        &\quad\Rightarrow\quad  n^+_t(C) \geq \tfrac14\, n_t(C)
        &\quad\Rightarrow\quad \rad^+_t(C) \leq \rad_t(C),\\
    n_t(C) < n_0
        &\quad\Rightarrow\quad  \rad_t(C) \geq 1
        &\quad\Rightarrow\quad  \abs{V^+(C)-V^+_t(C)} \leq \rad_t(C).
\end{array}$$
Therefore, once Equations \eqref{eqn:clean-V-almost} and \eqref{eq:clean-auxiliary} hold, we have
     $\abs{V^+(C)-V^+_t(C)} \leq \rad_t(C)$.
This completes the proof of \refeq{eqn:clean-V}.
\end{proof}

\subsection{Analysis of a clean execution}

The rest of the analysis focuses on a clean execution. Recall that $C_t$ is the cell chosen by the algorithm in round $t$.

\begin{claim}\label{cl:index-LB}
    In any clean execution, $I(C_t)\geq \OPT(\Xcand)$ for each round $t$.
\end{claim}
\begin{proof}
Fix round $t$, and let $x^*$ be any candidate contract. By Invariant~\ref{cl:zooming-invariants}(b), there exists an active cell, call it $C^*_t$, which contains $x^*$.

We claim  that $I_t(C^*_t) \geq U(x^*)$. We consider two cases, depending on whether $C^*_t$ is atomic. If $C^*_t$ is atomic then the anchor is unique, so $U(C^*_t) = U(x^*)$, and
    $I_t(C^*_t)  \geq U(x^*)$
by the clean execution. If $C^*_t$ is composite then
\begin{align*}
I_t(C^*_t)
    & \geq U(C^*_t) + \vw(C^*_t) & \text{by clean execution}  \\
    & \geq U(C^*_t) + \width(C^*_t) & \text{by Lemma~\ref{lm:virtual_width}}  \\
    & \geq U(x^*) & \text{by definition of $\width$, since $x^*\in C^*_t$}.
 \end{align*}
We have proved that
    $I_t(C^*_t) \geq U(x^*)$.
Now, by the selection rule we have
 $I_t(C_t) \geq I_t(C^*_t) \geq U(x^*)$.
Since this holds for any candidate contract $x^*$, the claim follows.
% A.S.: We do not need to worry whether there exists
%      a utility-maximizing candidate contract.
\end{proof}

\begin{claim}\label{cl:index-UB}
In any clean execution, for each round $t$, the index $I_t(C_t)$ is upper-bounded as follows:
\begin{OneLiners}
\item[(a)] if $C_t$ is atomic then
    $I(C_t)\leq U(C_t) + 2\,\rad_t(C_t)$.
\item[(b)] if $C_t$ is composite then
    $I(C_t)\leq U(x) + O(\rad_t(C_t))$
for each contract $x\in C_t$.
\end{OneLiners}
\end{claim}

\begin{proof}
Fix round $t$. Part (a) follows because $I_t(C_t) = U_t(C_t)+\rad_t(C_t)$ by definition of the index, and
    $U_t(C_t) \leq U(C_t) + \rad_t(C_t)$
by clean execution.

For part (b), fix a contract $x\in C_t$. Then:
\begin{align}
U_t(C_t)
    &\leq U(C_t) + \rad_t(C_t)
        & \text{by clean execution} \nonumber \\
    &\leq U(x) + \width(C_t) + \rad_t(C_t)
        & \text{by definition of width} \nonumber \\
    &\leq U(x) + \vw(C_t)  + \rad_t(C_t)
        & \text{by Lemma~\ref{lm:virtual_width}} \nonumber \\
    &\leq U(x) + W_t(C_t) + 5\,\rad_t(C_t) & \text{by clean execution}.
        \label{eq:lem:badness-2} \\
I_t(C_t)
    &= U_t(C_t) + W_t(C_t) + 5\,\rad_t(C_t)
        & \text{by definition of index}   \nonumber \\
	&\leq U(x) + 2\,W_t(C_t) + 10\,\rad_t(C_t)
        & \text{by \refeq{eq:lem:badness-2}} \nonumber \\
    &\leq U(x) + 20\,\rad_t(C_t)
        & \text{by Invariant~\ref{cl:zooming-invariants}(c)}. \nonumber \qquad\qquad\quad \qedhere
\end{align}
\end{proof}

For each relevant cell $C$, define badness $\Delta(C)$ as follows. If $C$ is composite,
    $\Delta(C) = \sup_{x\in C}\, \Delta(x)$
is the maximal badness among all contracts in $C$. If $C$ is atomic and $x\in C$ is the unique candidate contract in $C$, then $\Delta(C) = \Delta(x)$.

\begin{claim}\label{cl:badness}
In any clean execution,
    $\Delta(C) \leq O(\rad_t(C))$
for each round $t$ and each active cell $C$.
\end{claim}

\begin{proof}
By Claims~\ref{cl:index-LB} and~\ref{cl:index-UB},
    $\Delta(C_t) \leq O(\rad_t(C_t))$
for each round $t$.
Fix round $t$ and let $C$ be an active cell in this round. If $C$ has never be selected before round $t$, the claim is trivially true. Else, let $s$ be the most recent round before $t$ when $C$ is selected by the algorithm. Then
    $\Delta(C) \leq O(\rad_s(C))$.
The claim follows since $\rad_s(C) = \rad_t(C)$.
\end{proof}

\begin{claim}\label{cl:badness2}
In a clean execution,
each cell $C$ is selected $\leq O(\log T / (\Delta(C))^2)$ times.
\end{claim}

\begin{proof}
By Claim~\ref{cl:badness},
    $\Delta(C) \leq O(\rad_T(C))$.
The claim follows from the definition of $\rad_T$ in \refeq{eq:conf-rad}.
\end{proof}

%\begin{proof}[Proof of Theorem~\ref{thm:regret}]
Let $n(x)$ and $n(C)$ be the number of times contract $x$ and cell $C$, respectively, are chosen by the algorithm. Then regret of the algorithm is
\begin{align}\label{eq:regret-sum}
R(T|\Xcand) = \textstyle \sum_{x\in X} n(x) \; \Delta(x)
    \leq \sum_{\text{cells $C$}}\; n(C) \, \Delta(C).
\end{align}
The next result (Lemma~\ref{lm:clean-execution}) upper-bounds the right-hand side of \refeq{eq:regret-sum} for a clean execution. By Lemma~\ref{lem:clean_prob}, this suffices to complete the proof of  Theorem~\ref{thm:regret}

\begin{lemma}\label{lm:clean-execution}
Consider a clean execution of \zooming. For any $\delta\in (0,1)$,
\begin{align*}
\textstyle
\sum_{\text{cells $C$}}\; n(C) \, \Delta(C)
\leq
    \delta T + O(\log T)
	       \sum_{\eps=2^{-j}\geq\delta:\; j\in \N}\; \frac{|\mF_\eps(X_{2\eps})|}{\eps}.
\end{align*}
\end{lemma}

\newcommand{\parent}{C_{\mathtt{p}}}

The proof of Lemma~\ref{lm:clean-execution} relies on some simple properties of $\Delta(\cdot)$, stated below.

\begin{claim}\label{cl:badness-props}
Consider two relevant cells $C\subset \parent$. Then:
\begin{OneLiners}
\item[(a)] $\Delta(C)\leq \Delta(\parent)$.
\item[(b)] If $\Delta(C)\leq \eps$ for some $\eps>0$, then $C$ overlaps with $X_\eps$.
\end{OneLiners}
\end{claim}

\begin{proof}
To prove part (a), one needs to consider two cases, depending on whether cell $\parent$ is composite. If it is, the claim follows trivially. If $\parent$ is atomic, then
$C$ is atomic, too, and so $\Delta(C) = \Delta(\parent)= \Delta(x)$, where $x$ is the unique candidate contract in $\parent$.

For part (b), there exists a candidate contract $x\in C$. It is easy to see that
    $\Delta(x)\leq \Delta(C)$
(again, consider two cases, depending on whether $C$ is composite.)
So, $x\in X_\eps$.
\end{proof}

\begin{proof}[Proof of Lemma~\ref{lm:clean-execution}]
Let $\Sigma$ denote the sum in question. Let $\mA^*$ be the collection of all cells ever activated by the algorithm. Among such cells, consider those with badness on the order of $\eps$:
\[
	\mG_\eps := \left\{\;  C\in\mA^*: \Delta(C) \in [\eps, 2\eps) \; \right\}.
\]
By Claim~\ref{cl:badness2}, the algorithm chooses each cell $C\in
\mG_\eps$ at most $O(\log T/ \epsilon^2)$ times, so
    $n(C)\, \Delta(C) \leq O(\log T / \epsilon)$.

Fix some $\delta\in (0,1)$ and observe that all cells $C$ with $\Delta(C)\leq \delta$ contribute at most $\delta T$ to $\Sigma$. Therefore it suffices to focus on
   $\mG_\eps$,  $\eps\geq \delta/2$.
It follows that
\begin{align}\label{eq:regret-proof-almost-done}
	\Sigma &\leq  \delta T + O(\log T)
       \textstyle \sum_{\eps = 2^{-i} \geq \delta/2} \; \frac{|\mG_\eps|}{\eps}.
\end{align}

We bound $|\mG_\eps|$ as follows. Consider a cell $C\in \mG_\eps$. The
cell is called a \emph{leaf} if it is never zoomed in on (i.e., removed from the active
set) by the algorithm. If $C$ is activated in the round when cell
$\parent$ is zoomed in on, $\parent$ is called the \emph{parent} of
$C$. We consider two cases, depending on whether or not $C$ is a leaf.
\begin{itemize}

\item[(i)] Assume cell $C$ is not a leaf. Since
    $\Delta(C) < 2\eps$,
$C$ overlaps with $X_{2\eps}$ by Claim~\ref{cl:badness-props}(b). Note
that $C$ is zoomed in on in some round, say in round $t-1$. Then
\begin{align*}
5\,\rad_t(C)
    &\leq W_t(C) & \text{by the zooming rule} \\
    &\leq \vw(C) + 4\,\rad_t(C) & \text{by clean execution},
\end{align*}
so
    $\rad_t(C) \leq \vw(C)$.
Therefore, using Claim~\ref{cl:badness}, we have
    $$\eps \leq \Delta(C) \leq O(\rad_t(C)) \leq O(\vw(C)).$$
It follows that
	$C \in \mF_{\Omega(\eps)}(X_{2\eps})$.

\item[(ii)] Assume cell $C$ is a leaf. Let $\parent$ be the parent of $C$. Since
    $C \subset \parent$,
we have
    $\Delta(C)\leq \Delta(\parent)$ by Claim~\ref{cl:badness-props}(a).
Therefore, invoking case (i), we have
    $$\eps \leq \Delta(C) \leq \Delta(\parent) \leq O(\vw(\parent)).$$
Since $\Delta(C)<2\eps$, $C$ overlaps with $X_{2\eps}$ by Claim~\ref{cl:badness-props}(b), and therefore so does $\parent$. It follows that
    $\parent\in \mF_{\Omega(\eps)}(X_{2\eps})$.
\end{itemize}

Combing these two cases, it follows that
    $|\mG_\eps| \leq (2^m+1) \abs{\mF_{\Omega(\eps)}(X_{2\eps})}$.
Plugging this into \eqref{eq:regret-proof-almost-done} and making an appropriate substitution $\eps \to \Theta(\eps)$ to simplify the resulting expression, we obtain the regret bound in Theorem~\ref{thm:regret}
\end{proof}

\section{Simulations}
\label{sec:simulation}

\newcommand{\discsize}{\ensuremath{\minSize}\xspace} % discretization granularity in simulations
\newcommand{\marketUnif}{Uniform Worker Market\xspace}
\newcommand{\marketHomo}{Homogeneous Worker Market\xspace}
\newcommand{\marketTwoType}{Two-Type Market\xspace}

We evaluate the performance of \zooming through simulations. \zooming is compared with two versions of \NonAdaptive that use, respectively, two standard bandit algorithms: \UCB \citep{bandits-ucb1} and Thompson Sampling \citep{Thompson-1933} (with Gaussian priors). For both \UCB and \zooming, we replace the logarithmic confidence terms with small constants. (We find such changes beneficial in practice, for both algorithms; this observation is consistent with prior work \citep{RBA-icml08,ZoomingRBA-icml10}.)
All three algorithms are run with $\Xcand = \Xcand(\minSize)$, where $\minSize>0$ is the granularity of the discretization.

\OMIT{In \zooming, we pick the anchors $x^+(C)$ and $x^-(C)$ of a composite cell $C$ to be the two candidate contracts within $C$ in which all increments payments are maximal/minimal.\cj{added one sentence about how we choose anchors.}}

\xhdr{Setup.}
We consider a version of the \hlexample, as described in Section~\ref{sec:hlexample}. We set the requester's values to $V(\high)=1$ and $V(\low)=.3$. The probability of obtaining \high outcome given \high effort is set to $\prH=.8$. Thus, the worker's type is characterized by the cost $\cH$ for \high effort. We consider three supply distributions:
\begin{itemize}
\item \emph{\marketUnif}: $\cH$ is uniformly distributed on $[0,1]$.
\item \emph{\marketHomo}: $\cH$ is the same for every worker.
\item \emph{\marketTwoType}: $\cH$ is uniformly distributed over two values, $\cH'$ and $\cH''$.
\end{itemize}

These first two markets represent the extreme cases when workers are extremely homogeneous or extremely diverse, and the third market is one way to represent the middle ground. For each market, we run each algorithm 100 times. For \marketHomo, $\cH$ is drawn uniformly at random from $[0,1]$ for each run. For \marketTwoType, $\cH'$ and $\cH''$ are drawn independently and uniformly from $[0,1]$ on each run.

\xhdr{Overview of the results.}
%The plots can be found in \full.
%Appendix~\ref{app:simulation}.
Across all simulations, \zooming performs comparably to or better than \NonAdaptive. 
%In particular, 
Its performance does not appear to suffer from large ``hidden constants" that appear in the analysis. We find that \zooming converges faster than \NonAdaptive when \discsize is near-optimal or smaller; this is consistent with the intuition that \zooming focuses on exploring the more promising regions. When \discsize is large, \zooming converges slower than \NonAdaptive, but eventually achieves the same performance. Further, we find that \zooming with small \discsize performs well compared to \NonAdaptive with larger  \discsize: not much worse initially, and much better eventually.

Our simulations suggest that \emph{if} time horizon $T$ is known in advance \emph{and} one can tune $\minSize$ to $T$, then \NonAdaptive can achieve similar performance as \zooming.  However, in real applications approximately optimal $\minSize$ may be difficult to compute, and the $T$ may not be known in advance.

%\section{Simulation results}
%\label{app:simulation}

\xhdr{Detailed results.}
Recall that in both \UCB and \zooming, the logarithmic confidence terms are replaced with small constants. For \UCB, the confidence term is $1$, so that if a given arm $a$ has been played $n_a$ times, its index is simply the average reward plus $1/\sqrt{n_a}$. For \zooming, we set $\rad_t(\cdot)=1$ in the selection rule, and $\rad_t(\cdot)=.6$ in the zooming rule. For both algorithms, we tried several values and picked those that performed well across all three markets; we found that the performance of both algorithms is not very sensitive to the particular choice of these constants, as long as they are on the order of $1$.

\newcommand{\AveU}{\widehat{U}}

%\xhdr{Results.}
For each algorithm, we compute the time-averaged cumulative utility after $T$ rounds given granularity $\minSize$, denote it $\AveU(T,\minSize)$, for various values of $T$ and $\minSize$.

First, we fix the time horizon $T$ to 5K rounds, and study how $\AveU(T,\minSize)$ changes with $\minSize$ (see Figure~\ref{fig:eps}). We observe that \zooming either matches or outperforms both versions of \NonAdaptive, across all markets and all values of $\minSize$. \zooming has a huge advantage when $\minSize$ is small.

\begin{figure}[ht]
	\centering
   \subfigure[\marketUnif]{
		\includegraphics[width=0.48\columnwidth, keepaspectratio]{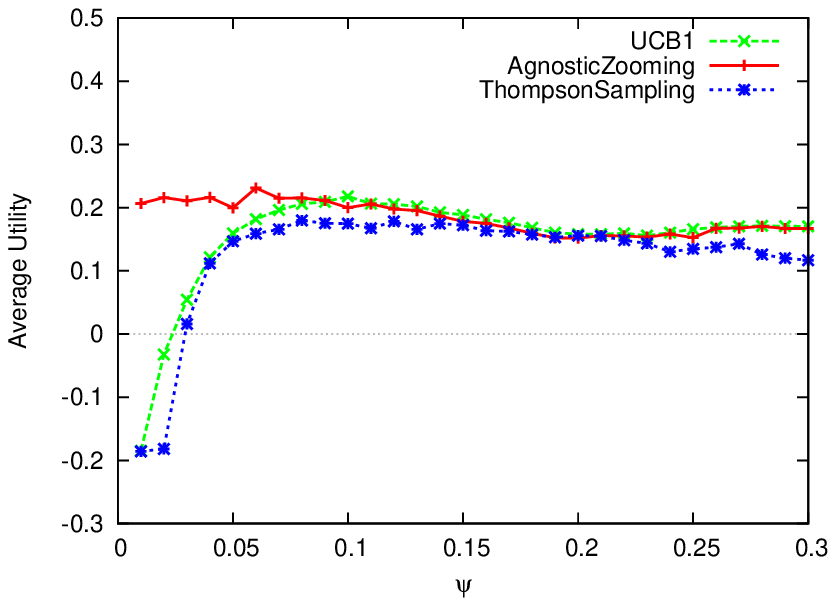}
		\label{fig:uniform}
	}
   \subfigure[\marketHomo]{
		\includegraphics[width=0.48\columnwidth, keepaspectratio]{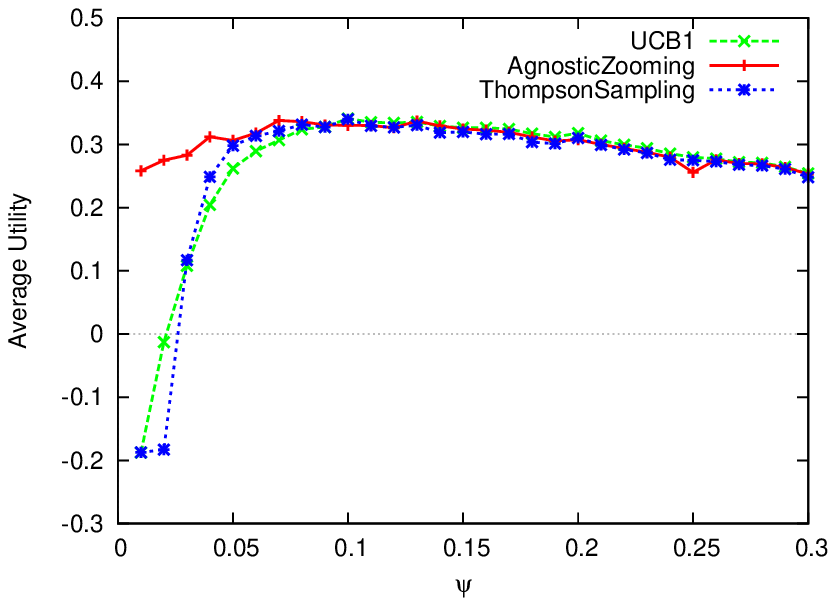}
		\label{fig:homogeneous}
	}
   \subfigure[\marketTwoType]{
		\includegraphics[width=0.48\columnwidth, keepaspectratio]{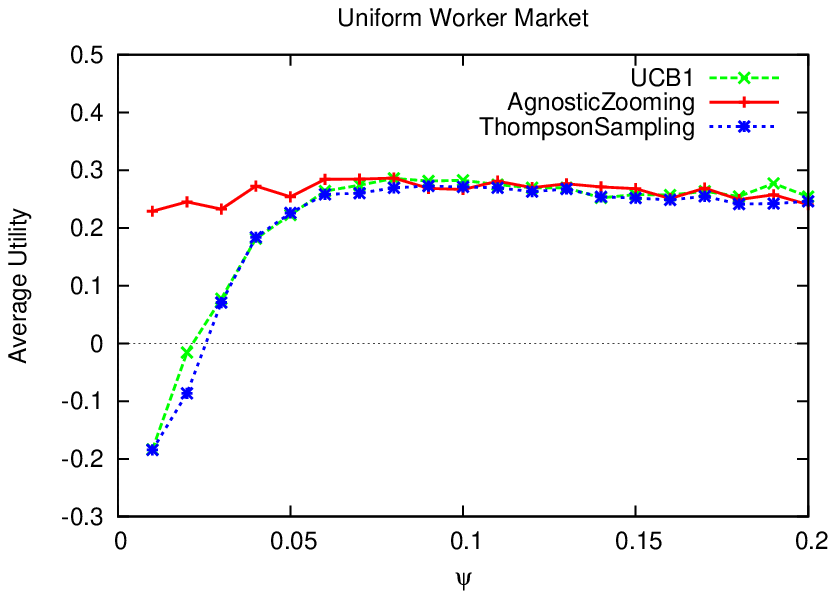}
		\label{fig:two}
	}
	\caption{The requester's payoff after $5,000$ rounds vs. the choice of initial discretization \discsize.}
	\label{fig:eps}
\end{figure}

Second, we study how the three algorithms perform over time. Specifically, we plot $\AveU(T,\minSize)$ vs. $T$, for three values of $\minSize$, namely $0.02$, $0.8$, and $0.2$. Since setting to $0.08$ is close to optimal in our examples, these values of $\minSize$ represent, resp., too small, adequate, and too large. The results are shown in Figure~\ref{fig:over-time}. We find that \zooming converges faster than \NonAdaptive when \discsize is adequate or small; this is consistent with the intuition that \zooming focuses on exploring the more promising regions. When \discsize is large, \zooming converges slower than \NonAdaptive, but eventually achieves the same performance.

\begin{figure}[p]
	\centering
   \subfigure{%[Uniform Workers;\discsize = 0.02]{
		\includegraphics[width=0.32\columnwidth, keepaspectratio]{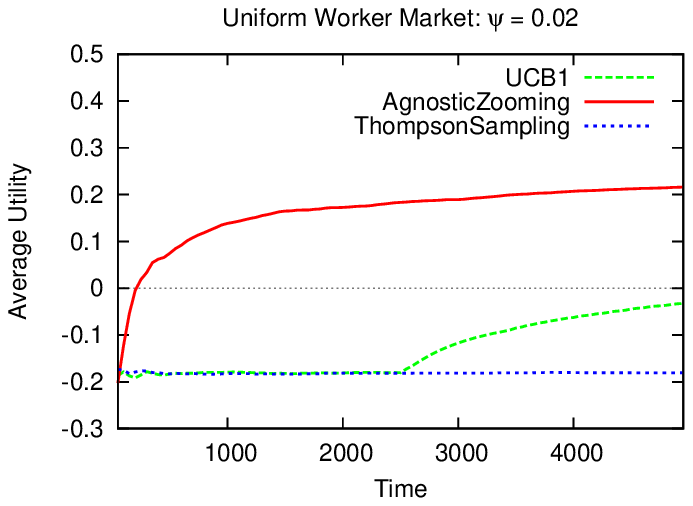}
		\label{fig:uniform0.02}
	}
   \subfigure{%[Uniform Workers; \discsize = 0.08]{
		\includegraphics[width=0.32\columnwidth, keepaspectratio]{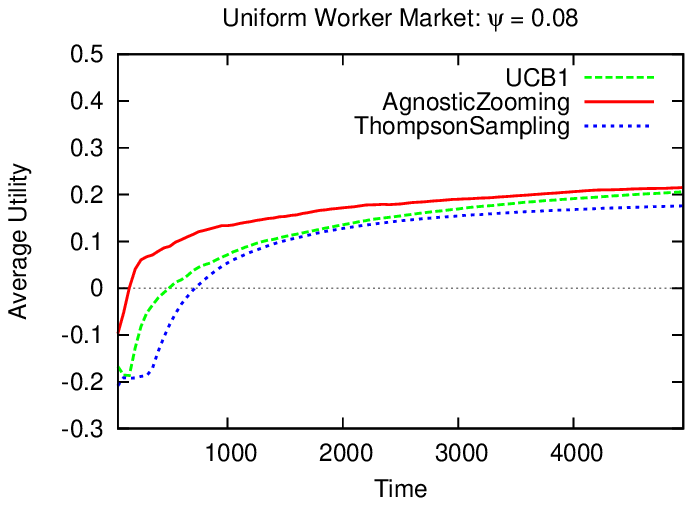}
		\label{fig:uniform0.08}
	}
   \subfigure{%[Uniform Workers.\protect\linebreak \discsize = 0.2]{
		\includegraphics[width=0.32\columnwidth, keepaspectratio]{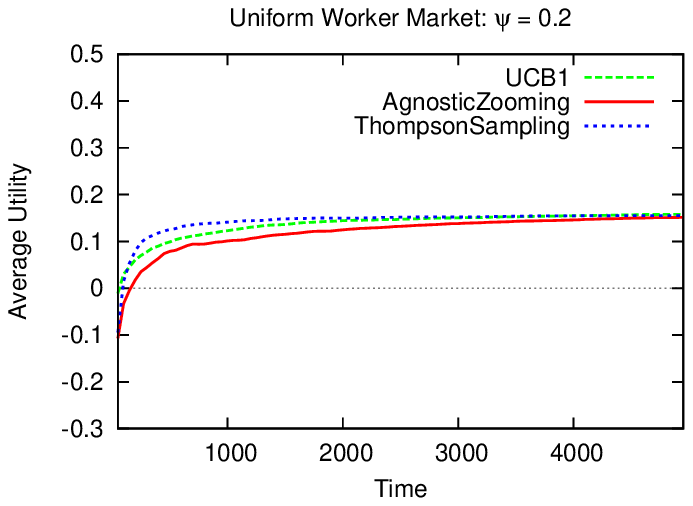}
		\label{fig:uniform0.2}
	}
   \subfigure{%[HomogeneousWorkers; \discsize = 0.02.]{
		\includegraphics[width=0.32\columnwidth, keepaspectratio]{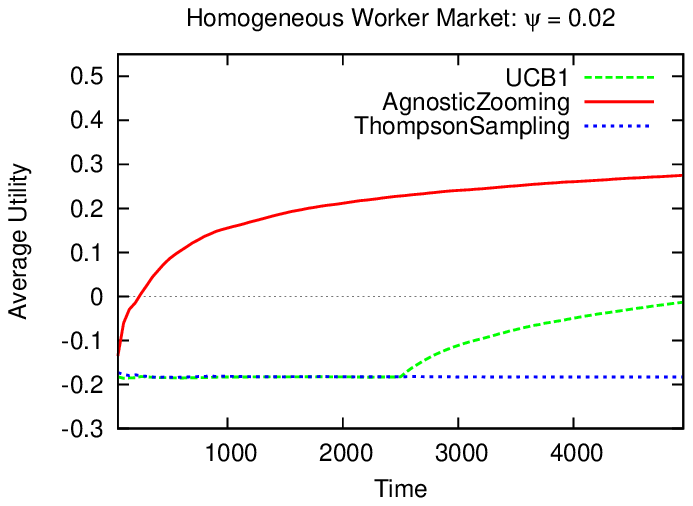}
		\label{fig:homogeneous0.02}
	}
   \subfigure{%[Homogeneous Workers; \discsize = 0.08.]{
		\includegraphics[width=0.32\columnwidth, keepaspectratio]{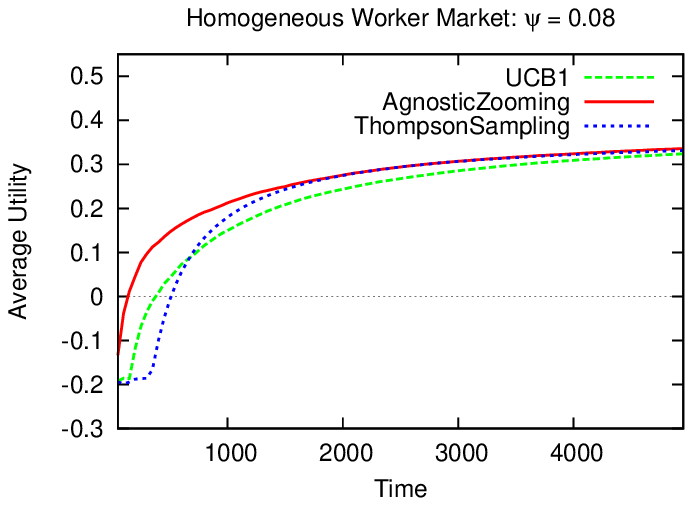}
		\label{fig:homogeneous0.08}
	}
   \subfigure{%[Homogeneous Workers; \discsize = 0.2.]{
		\includegraphics[width=0.32\columnwidth, keepaspectratio]{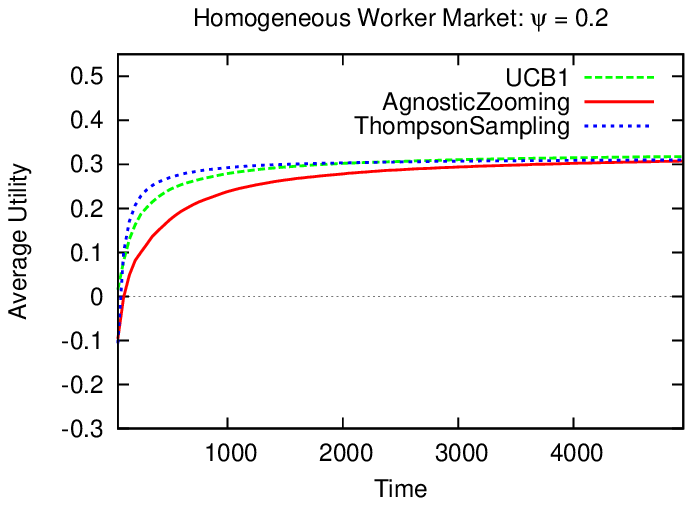}
		\label{fig:homogeneous0.2}
	}
   \subfigure{%[HomogeneousWorkers; \discsize = 0.02.]{
		\includegraphics[width=0.32\columnwidth, keepaspectratio]{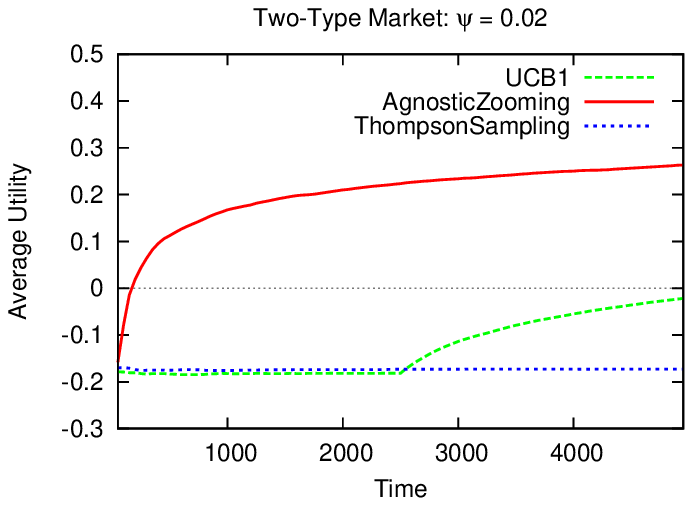}
		\label{fig:two0.02}
	}
   \subfigure{%[Homogeneous Workers; \discsize = 0.08.]{
		\includegraphics[width=0.32\columnwidth, keepaspectratio]{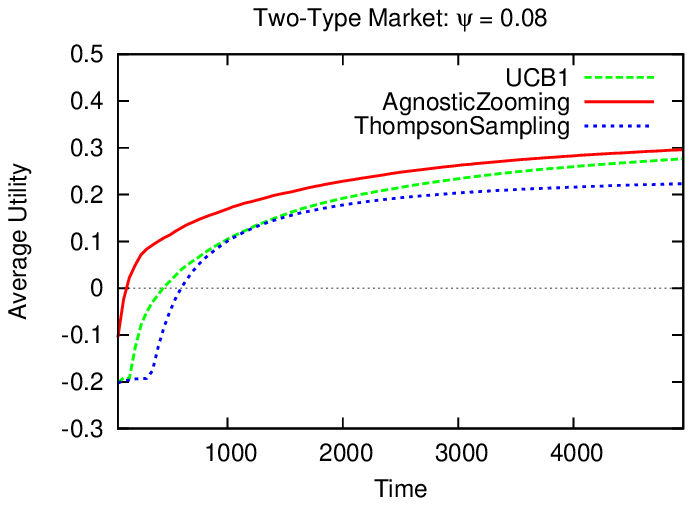}
		\label{fig:two0.08}
	}
   \subfigure{%[Homogeneous Workers; \discsize = 0.2.]{
		\includegraphics[width=0.32\columnwidth, keepaspectratio]{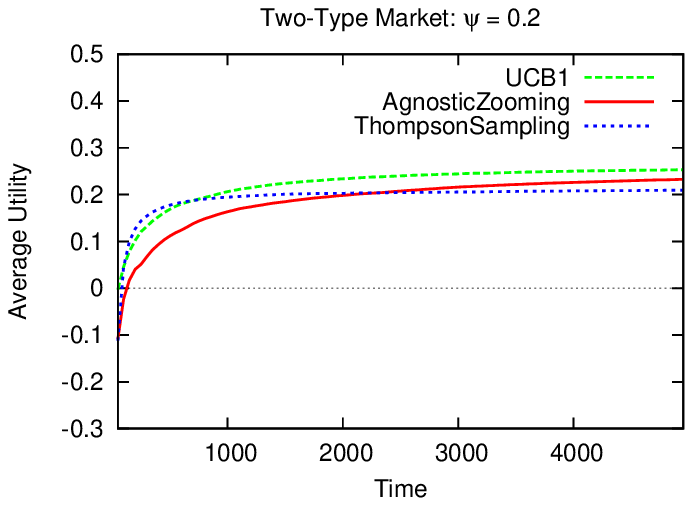}
		\label{fig:two0.2}
	}
	\caption{Algorithm performance over time under different discretization size.}
	\label{fig:over-time}
\end{figure}

\begin{figure}[p]
   \centering
	\begin{minipage}{.45\textwidth}
	\raggedleft
   \includegraphics[width=1\columnwidth, keepaspectratio]{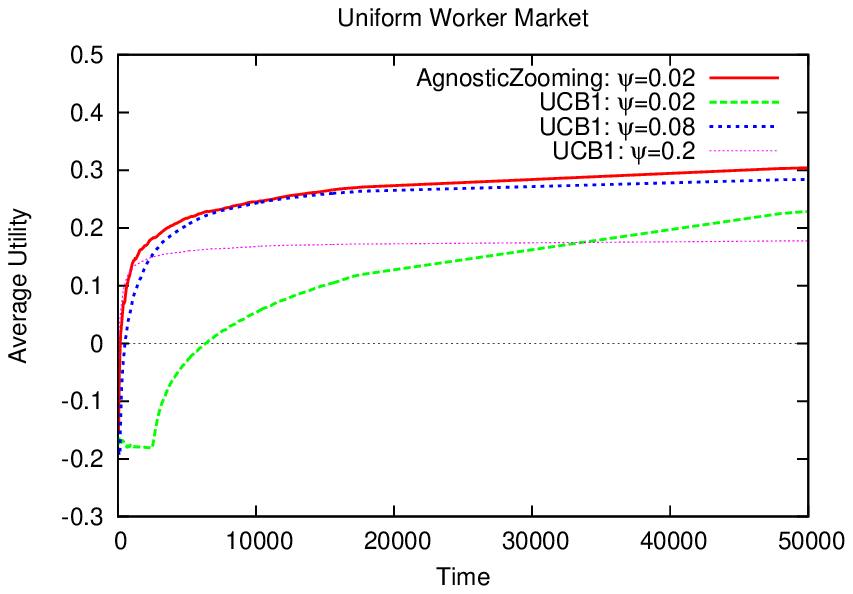}
   \caption{\zooming with small \discsize, compared to \NonAdaptive with three different values for \discsize. }
   \label{fig:uniform-comparison}
	\end{minipage}
   \begin{minipage}{.1\textwidth}
   \end{minipage}
   \begin{minipage}{.45\textwidth}
	\raggedright
	\includegraphics[width=1\columnwidth, keepaspectratio]{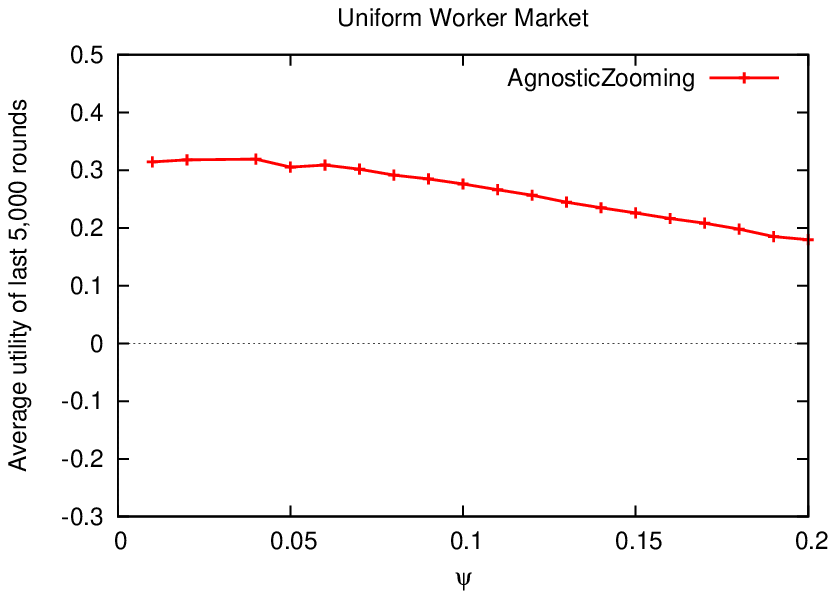}
	\caption{Average utility over the last 5K rounds in a 50K-round run of \zooming.}
	\label{fig:uniform-limit}
	\end{minipage}
\end{figure}

%\xhdr{Unknown Rounds $T$.}\\
%\xhdr{Running \zooming with small \discsize.}\\

Our simulations suggest that if time horizon $T$ is known in advance and one can optimize the $\minSize$ given this $T$, then \NonAdaptive can achieve similar performance as \zooming.  However, in real applications approximately optimal $\minSize$ may be difficult to calculate; further, the $T$ may be unknown in advance.

Third, we argue that \zooming performs well with a small \discsize: we compare its performance against that for \NonAdaptive with different values of \discsize. For each algorithm and each choice of $\minSize$, we plot $\AveU(T,\minSize)$ vs. $T$, see Figure~\ref{fig:uniform-comparison}.~
\footnote{We only show the results for \marketUnif; the results for \marketHomo are very similar. We only show the version of \NonAdaptive with \UCB, because in our experiments it performs better that Thompson Sampling.
(We conjecture that this is because we replaced the logarithmic confidence term in \UCB with $1$.)}
 We find that for small $T$, \zooming with small \discsize converges nearly as fast as \NonAdaptive with larger \discsize. When $T$ is large, \zooming with small \discsize converges to a better payoff than \NonAdaptive with larger \discsize.

Fourth, we confirm the intuition that $\OPT(\Xcand(\minSize))$ decreases with the granularity $\minSize$. To this end, we run \zooming for 50K rounds, and take the average utility over the last 5K rounds, see Figure~\ref{fig:uniform-limit}.

The standard errors in all plots are in the order of $0.001$ or less. (Note that each point is not only the average of $100$ runs but also the average of all previous rounds.)

\section{Application to \taskPricing}
\label{sec:pricing}

We discuss \emph{\taskPricing}, which can be seen as  the special case of \problem in which there is exactly one non-null outcome. We identify an important family of problem instances for which \zooming out-performs \NonAdaptive.

\xhdr{Some background.}
The \taskPricing problem, in its most basic version, is defined as follows. There is one \emph{principal} (buyer) who sequentially interacts with multiple \emph{agents} (sellers). In each round $t$, an agent arrives, with one item for sale. The principal offers price $p_t$ for this item, and the agent agrees to sell if and only if $p_t\geq c_t$, where $c_t\in[0,1]$ is the agent's private \emph{cost} for this item. The principal derives value $v$ for each item bought; his utility is the value from bought items minus the payment. The time horizon $T$ (the number of rounds) is known.  Each private cost $c_t$ is an independent sample from some fixed distribution, called the \emph{supply distribution}. We are interested in the \emph{prior-independent} version, where the supply distribution is not known to the principal. The algorithm's goal is to choose the offered prices $p_t$ so as to maximize the expected utility of the principal.

\TaskPricing can be seen as the special case of \problem in which there is exactly one non-null outcome (which corresponds to a sale). Indeed, in this special case there is exactly one non-null effort level $e$ without loss of generality (because any non-null effort levels deterministically lead to the non-null outcome).

One crucial simplification compared to the full generality of \problem is that the discretization error can now be easily bounded from above:
\footnote{Recall that $\Mesh(\minSize)$ denotes the set of all prices in $[0,1]$ that are integer multiples of a given $\minSize>0$; call this set the \emph{additive $\minSize$-mesh}.}
    $$\OPT(X) - \OPT(\Xcand(\minSize)) \leq \minSize \quad \text{for each $\minSize>0$}.$$

Worst-case regret bounds are implicit in prior work on dynamic inventory-pricing~\citep{Bobby-focs03}.%
\footnote{The algorithmic result for \taskPricing is an easy modification of the analysis in~\citet{Bobby-focs03} for dynamic inventory-pricing. The lower bound in in~\citet{Bobby-focs03} can also be ``translated" from dynamic inventory-pricing to \taskPricing without introducing any new ideas. We omit the details from this version.}
Let $\NonAdaptive(\minSize)$ denote algorithm $\NonAdaptive$ with
    $\Xcand = \Mesh(\minSize)$.
Then, by the analysis in~\citet{Bobby-focs03}, $\NonAdaptive(\minSize)$
achieves regret
    $ R(T) = \tildeO(\minSize T + \minSize^{-2})$.
This is optimized to
    $R(T) = \tildeO(T^{2/3})$
if and only if
    $\minSize = \tildeO(T^{-1/3})$.
Moreover, there is a matching lower bound:
    $R(T) =\Omega(T^{2/3})$
for any algorithm.

Further, it is a folklore result that $\NonAdaptive(\minSize)$ achieves regret $R(T) =\tildeO(T^{2/3})$ if and only if
    $\minSize = \tilde{\Theta}(T^{-1/3})$.
(We sketch a lower-bounding example in the proof of Lemma~\ref{lm:procurement-NonAdaptive-generalLB}\ignore{Section~\ref{subsec:taskPricing-LB}}, to make the paper more self-contained.)

\xhdr{Preliminaries.}
Each contract is summarized by a single number: the offered price $p$ for the non-null outcome. Let $F(p)$ be the probability of a worker accepting a task at price $p$, and let
    $U(p) = F(p)\,(v-p)$
be the corresponding expected utility of the algorithm.

Note that all contracts are trivially monotone and any optimal contract is bounded without loss of generality. It follows that
    $\OPT(X) = \sup_{p\geq 0} U(p)$,
the optimal expected utility over all possible prices.

A cell $C$ is just a price interval $C = [p,p'] \subset [0,1]$, and its virtual width is
\begin{align*}
\vw(C) =  \left( v\, F(p') - p\, F(p)  \right) - \left( v\, F(p) - p'\, F(p') \right).
\end{align*}

\xhdr{Our results: the general case.}
We will be using \zooming with $\Xcand =X$.

First, let us prove that this is a reasonable choice in the worst case: namely, that we achieve the optimal $\tildeO(T^{2/3})$ regret.

\begin{lemma}\label{lm:taskPricing-fixexT}
Consider the \taskPricing problem. \zooming with $\Xcand =X$ achieves regret $O(T^{2/3} \log T)$.
\end{lemma}

\begin{proof}[Proof Sketch]
Fix $\eps>0$. The key observation is that if $\vw(C)\geq \eps$ then either $p'-p\geq \tfrac{\eps}{4}$, or $F(p')-F(p) \geq \tfrac{\eps}{4}$. Call $C$ a \emph{red} cell if the former happens, and \emph{blue} cell otherwise. Therefore in any collection of mutually disjoint cells of virtual width $\geq \eps$ there can be at most $O(\freps)$ red cells and at most $O(\freps)$ blue cells, hence at most $O(\freps)$ cells total. It follows that there can be at most $O(\freps)$ active cells of virtual width $\geq \eps$.

So, in the notation of Theorem~\ref{thm:regret} we have $N_{\eps}(\cdot)\leq O(\freps)$. It follows that the width dimension is at most $1$, which in turn implies the desired regret bound.
\end{proof}

\OMIT {%%%%%%%%
Further, using essentially the same argument in conjunction with the correspondingly modified version of Theorem~\ref{thm:regret}, we can derive a regret bound $R(t)$ that that holds for any given time interval $[t_0,t_0+t]\subset [1,T]$. Such regret bounds will be called \emph{uniform}.

\begin{lemma}\label{lm:taskPricing-fixexT}
Consider the \taskPricing problem with a given time horizon $T$. \zooming with $\Xcand =X$ achieves regret
    $R(t) = O(t^{2/3} \log T)$
for any given time interval $[t_0,t_0+t]\subset [1,T]$.
\end{lemma}
} %%%%%%

\xhdr{Our results: ``nice" problem instances.}
We focus on problem instances with piecewise-uniform costs and bounded density. Formally, we say that an instance of \taskPricing has  \emph{$k$-piecewise-uniform costs} if the interval [0,1] is partitioned into $k\in \N$ sub-intervals such that the supply distribution is uniform on each sub-interval. A problem instance has \emph{$\lambda$-bounded density}, $\lambda\geq 1$ if the supply distribution has a probability density function almost everywhere, and the density is between $\tfrac{1}{\lambda}$ and  $\lambda$.  Using the full power of Theorem~\ref{thm:regret}, we obtain the following regret bound.

\begin{theorem}\label{thm:procurement}
Consider the \taskPricing problem with $k$-piecewise-uniform costs and $\lambda$-bounded density, for some absolute constants $k\in \N$ and $\lambda>1$.  \zooming with $\Xcand=X$ achieves regret
%    $$R(T) = \tildeO(\min( \lambda^{1/3}  T^{2/3},\, \lambda^{3/5} k^{2/5} T^{3/5})).$$
    $R(T) = \tildeO(T^{3/5})$.
\end{theorem}

\begin{proof}[Proof Sketch]
Since the supply distribution has density at most $\lambda$, it follows that $F(\cdot)$ is a Lipschitz-continuous function with Lipschitz constant $\lambda$.  It follows that each cell of virtual width at least $\eps$ has diameter at least $\Omega(\eps/\lambda)$, for any $\eps>0$. (Note that each ``cell" is now simply a sub-interval $[p,q]\subset [0,1]$, so its diameter is simply $q-p$.)

Second, we claim that $X_\eps$ is contained in a union of $k$ intervals of diameter $O(\sqrt{\eps \lambda})$. To see this, consider the partition of $[0,1]$ into $k$ subintervals such that the supply distribution has a uniform density on each subinterval. Let $[p_j,q_j]$ be the $j$-th subinterval. Let $p^*_j$ be the local optimum of $U(\cdot)$ on this subinterval, and let
    $X_{j,\eps} = \{ x\in [p_j,q_j]:\; U(p^*_j)-U(x) \leq \eps \}$.
Then
    $X_\eps \subset \cup_j X_{j,\eps}$.
We can show that
    $X_{j,\eps} \subset [p^*_j - \delta,\, p^*_j + \delta] $
for some $\delta = O(\sqrt{\eps \lambda})$.

Recall that $N_{\eps\beta_0}(X_\eps)$ is the number of feasible cells of virtual width at least $\eps \beta_0$ which overlap with $X_\eps$. It follows that $N_{\eps\beta_0}(X_\eps)$ is at most $k$ times the maximal number of feasible cells of diameter at least $\Omega(\eps/\lambda)$ that overlap with an interval of diameter $O(\sqrt{\eps \lambda})$. Therefore:
    $N_{\eps\beta_0}(X_\eps)
        = O(k \lambda^{3/2} \eps^{-1/2}\, \log \frac{1}{\eps})$.
Moreover, we have a less sophisticated upper bound on $N_{\eps\beta_0}(X_\eps)$: it is at most the number of feasible cell of diameter at least $\Omega(\eps/\lambda)$. So
    $N_{\eps\beta_0}(X_\eps)
        = O(\lambda/\eps) (\log \frac{1}{\eps})$.
The theorem follows by plugging both upper bounds on  $N_{\eps\beta_0}(X_\eps)$ into  \refeq{eq:thm:regret}.
\end{proof}

\xhdr{Comparison with \NonAdaptive.}
Consider
    $\NonAdaptive(\minSize_0)$, where $\minSize_0 = \tilde{\Theta}(T^{-1/3})$
is the granularity required for the optimal worst-case performance. Call a problem instance \emph{nice} if it has $2$-piecewise-uniform costs and $\lambda$-bounded density, for some sufficiently large absolute constant $\lambda$; say $\lambda=4$ for concreteness. We claim that \zooming outperforms $\NonAdaptive(\minSize_0)$ on the ``nice" problem instances.
\begin{lemma}\label{lm:procurement-NonAdaptive}
$\NonAdaptive(\minSize_0)$ achieves regret $R(T)=\Omega(T^{2/3})$ in the worst case over all ``nice" problem instances.
\end{lemma}

\begin{proof}[Proof Sketch]
Recall that for $k=2$ the supply distribution has density $\lambda_1$ on interval $[0,p_0]$, and density $\lambda_2$ on interval $[p_0,1]$, for some numbers $\lambda_1, \lambda_2,p_0$. We pick $p_0$ so that it is sufficiently far from any point in $\Mesh(\minSize_0)$. Note that the function $U(\cdot)$ is a parabola on each of the two intervals. We adjust the densities so that $U(\cdot)$ achieves its maximum at $p_0$, and the maximum of either of the two parabolas is sufficiently far from $p_0$. Then the discretization error of $\Mesh(\minSize_0)$ is at least $\Omega(\minSize_0)$, which implies regret $\Omega(\minSize_0 T)$.
\end{proof}

\xhdr{Lower bound for \NonAdaptive.}
We provide a specific lower-bounding example for the worst-case performance of $\NonAdaptive(\minSize)$, for an arbitrary $\minSize>0$. Let $\mF$ be the family of all problem instances with $k$-piecewise-uniform costs and $\lambda$-bounded density, for all $k\in\N$ and $\lambda=4$.

\begin{lemma}\label{lm:procurement-NonAdaptive-generalLB}
Let $R_\minSize(T)$ be the maximal regret of $\NonAdaptive(\minSize)$ over all problem instances in $\mF$. Then
    $R_\minSize(T) = \Omega(\minSize T + \sqrt{T/\minSize}) \geq \Omega(T^{2/3})$.
\end{lemma}

\begin{proof}[Proof Sketch]
For piecewise-uniform costs, we have $F(0)=0$ and $F(p)=1$. Assume that the principal derives value $v=1$ from each item. Then the expected utility from price $p$ is
    $U(p) = F(p) (1-p)$.

Fix $\minSize>0$. Use the following problem instance. Let
    $\mP_\delta = [\tfrac25,\tfrac35] \cap \{ 4j\minSize+\delta:\, j\in\N\}$.
Set $U(p)=\tfrac14$ for each $p\in \mP_0$. Further, pick some
    $p^*\in \mP_{\minSize/2}$
and set $U(p^*) = \tfrac14+\Omega(\minSize)$.
This defines $F(p)$ for $p\in \mP \cup \{0,1,p^*\}$. For the rest of the prices, define $F(\cdot)$ via linear interpolation. This completes the description of the problem instance.

We show that $X_\minSize$ consists of $N = \Omega(\tfrac{1}{\minSize})$ candidate contracts. Therefore, using standard lower-bounding arguments for MAB, we obtain
    $R(T|\Xcand) \geq \Omega(\sqrt{TN}) = \Omega(\sqrt{T/\minSize}) $.
Further, we show that the discretization error is at least $\Omega(\minSize)$, implying that
    $R(T)\geq R(T|\Xcand)+ \Omega(\minSize T)$.
\end{proof}
% appendix on "dynamic procurement" = "dynamic task pricing"
\section{Related work}
\label{sec:related-work}

This paper is related to three different areas: contract theory, market design for crowdsourcing, and online decision problems. Below we outline connections to each of these areas.

%\subsection{Contract theory}
%\label{sec:relatedcontract}

\xhdr{Contract theory.}
Our model can be viewed as an extension of the classic principal-agent model from contract theory~\citep{LM02}.  In the most basic version of the classic model, a single principal interacts with a single agent whose type (specified by a cost function and production function, as described in Section~\ref{sec:setting}) is generally assumed to be known.  The principal specifies a contract mapping outcomes to payments that the principal commits to make to the agent.  The agent then chooses an action (i.e., effort level) that stochastically results in an outcome in order to maximize his expected utility given the contract.  The principal observes the outcome, but cannot directly observe the agent's effort level, creating a \emph{moral hazard} problem.  The goal of the principal is to design a contract to maximize her own expected utility, which is the difference between the utility she receives from the outcome and the payment she makes.  This maximization can be written as a constrained optimization problem, and it can be shown that linear contracts are optimal.

The \emph{adverse selection} variation of the principal-agent problem relaxes the assumption that the agent's type is known.  Most existing literature on the principal-agent problem with adverse selection focuses on applying the revelation principle~\citep{LM02}.  In this setting, the principal offers a menu of contracts, and the contract chosen by the agent reveals the agent's type.  The problem of selecting a menu of contracts that maximizes the principal's expected utility can again be formulated as a constrained optimization.

\ignore{
The \emph{limited liability} variation of the principal-agent problem considers the problem under the additional restriction that payments must lie in a certain range, capturing settings in which payments must be non-negative or satisfy some other minimum wage~\citep{LM02,JKS08}. \jenn{Not sure what we want to say about this setting.  The \citet{JKS08} result is for risk averse agents and so doesn't directly apply.  They cite some papers that consider limited liability with risk neutral agents so maybe we should cite one of those instead? What are the important results?}
}

%   \item Hybrid Setting\\
%   Similar to adverse selection but need more incentive constraints.

% We looked at this paper again and couldn't figure out why it was different/special/relevant...
\ignore{
\citet{GM13} examine a setting with both moral hazard and adverse selection.  They consider only binary action spaces and outcome spaces and assume that a distribution over agent types is known.  They focus on characterizing the optimal solution and decrease the search space.
}

Our work differs from the classic setting in that we consider a principal interacting with multiple agents, and the principal may adjust her contract over time in an online manner.  Several other authors have considered extensions of the classic model to multiple agents.
\citet{LV02} show that with multiple agents it is optimal to set individual linear contracts for each agent rather than a single uniform contract for all agents, but offer a variety of descriptive explanations for why it is more common to see uniform contracts in practice.
\citet{BFN06} consider a setting in which one principal interacts with multiple agents, but observes only a single outcome which is a function of all agents' effort levels.
\citet{MND12} consider a variant in which the algorithm must decide both how to set a uniform contract for many agents and how to select a
subset of agents to hire.

Alternative online versions of the problem have been considered in the literature as well.  In dynamic principal agent problem~\citep{Sannikov08, William09, Sannikov12}, a single principal interacts with a single agent repeatedly over a period of time. The agent can choose to exert different effort at different time, and the outcome at time $t$ is a function of  all the efforts exerts by the agent before $t$. The principal cannot observe the agent's efforts but can observe the outcome. The goal of the principal is to design an optimal contract over time to maximize his payoff. Our work is different from this line of work since we consider the setting with multiple agents with different, unknown types. Our algorithm needs to learn the distribution of agent types and design an optimal contract accordingly.

\citet{CG06} studies the online principal agent problem with a similar setting to ours. However, they focus on empirically comparing different online algorithms, including bandit approaches with uniform discretization, gradient ascent, and Bayesian update approaches to the problem. Our goal is to provide an algorithm with nice theoretical guarantees.

\citet{BK13} studies the setting when the outcome is unverifiable. To address this issue, they propose to assign a bundle of tasks to each worker. To verify the outcome, each task in the bundle is chosen as a verifiable task with some non-trivial probability. A verifiable task can either be a gold standard task with known answer or a task which is assigned to multiple workers for verification. The payment for a task bundle is then conditional only on the outcome of verified tasks. In our setting, we assume the task outcome is verifiable. We can relax this assumption by adopting similar approaches.

%
%the output is a function of all the past efforts (with noise.)

% We only have one principal so this seems not directly relevant
\ignore{
\subsection{Multiple Principals and Multiple Agents}
\begin{itemize}
   \item ``The Principal-Agent Matching Market'', Kaniska Dam and David Perez-Castrillo, The B.E. Journal of Theoretical Economics, 2006 \cite{DP06}
\end{itemize}
}

%\subsection{Incentives in crowdsourcing systems}
%\label{sec:relatedcrowd}

\xhdr{Incentives in crowdsourcing systems.}
Researchers have recently begun to examine the design of incentive mechanisms to encourage high-quality work in crowdsourcing systems. \cite{JCP12} explore ways in which to award virtual points to users in online question-and-answer forums to improve the quality of answers. \cite{GH11,GH13} and \cite{GM11} study how to distribute user generated content (e.g., Youtube videos) to users to encourage the production of high-quality internet content by people who are motivated by attention. \cite{HZVv12} and \cite{Zv12} consider the design of two-sided reputation systems to encourage good behavior from both workers and requesters in crowdsourcing markets.  While we also consider crowdsourcing markets, our work differs in that it focuses on how to design contracts, perhaps the most natural incentive scheme, to incentivize workers to exert effort.

The problem closest to ours which has been studied in the context of crowdsourcing systems is the online task pricing problem in which a requester has an unlimited supply of tasks to be completed and a budget $B$ to spend on them~\citep{DynProcurement-ec12,SM13}. Workers with private costs arrive online, and the requester sets a single price for each arriving worker. The goal is to learn the optimal single fixed price over time. Our work can be viewed as a generalization of the task pricing problem, which is a special case of our setting with the number of non-null outcomes $m$ fixed at 1.

There has also been empirical work examining how workers' behavior varies based on the financial incentives offered in crowdsourcing markets. \cite{MW09} study how workers react to changes of performance-independent financial incentives. In their study, increasing financial incentives increases the number of tasks workers complete, but not the quality of their output. \cite{YCS13} provide a potential explanation for this phenomenon using the concept of ``anchoring effect'': a worker's cost for completing a task is influenced by the first price the worker sees for this task.
\cite{HC10} run experiments to estimate workers' reservation wage for completing tasks. They show that many workers respond rationally to offered contracts, whereas some of the workers appeared to have some ``target payment'' in mind.

Some recent research studies the effects of performance-based payments (PBPs). \cite{Harris11} runs MTurk experiments on resume screening, where workers can get a bonus if they perform well. He concludes that the quality of work is better with PBPs than with uniform payments. \cite{YCS13} show that varying the magnitude of the bonus does not have much effect in certain settings. \citet{PBPs-www15} perform a more comprehensive set of experiments aimed at determining whether, when, and why PBPs increase the quality of submitted work. Their results suggest that PBPs can increase quality on tasks for which increased time or effort leads to higher quality work. Their results also suggest that workers may interpret a contract as performance-based even if it is not stated as such (since requesters always have the option to reject work).  Based on this evidence, they propose a new model of worker behavior that extends the principal-agent model to explicitly reflect workers' subjective beliefs about their likelihood of being paid.

\ignore{
\cj{TODO: add a paragraph of performance-contingent pricing as described in Jenn's email. Also, add the work of 99design in HCOMP.}
\cite{Harris11, WIP13, Mao+13, Araujo13}
}

Overall, previous empirical work demonstrates that workers in crowdsourcing markets do respond to the change of financial incentives, but that their behavior does not always follow the traditional rational-worker model --- similar to people in any real-world market. In our work, we start our analysis with the rational-worker assumption ubiquitous in economic theory, but demonstrate that our results can still hold without these assumptions as long as the collective worker behavior satisfies some natural properties (namely, as long as Lemma~\ref{lm:virtual_width} holds). We note that our results hold under the generalized worker model proposed by \citet{PBPs-www15}, which is consistent with their experimental evidence as discussed above.

%\subsection{Online decision problems}
%\label{sec:relatedbandits}

\xhdr{Sequential decision problems.}
In sequential decision problems, an algorithm makes sequential decisions over time.Two directions that are relevant to this paper are multi-armed bandits (MAB) and dynamic pricing.
%see \citet{CesaBL-book} for more general background.

MAB have been studied since 1933~\citep{Thompson-1933} in Operations Research, Economics, and several branches of Computer Science including machine learning, theoretical computer science, AI, and algorithmic economics. A survey of prior work on MAB is beyond the scope of this paper; the reader is encouraged to refer to \citet{CesaBL-book} or \citet{Bubeck-survey12} for background on prior-independent MAB, and to \citet{Gittins-book11} for background on Bayesian MAB. Below we briefly discuss the lines of work on MAB that are directly relevant to our paper.

Our setting can be modeled as prior-independent MAB with stochastic rewards: the reward of a given arm $i$ is an i.i.d. sample of some time-invariant distribution, and neither this distribution nor a Bayesian prior on it are known to the algorithm. The basic formulation (with a small number of arms) is well-understood \citep{Lai-Robbins-85,bandits-ucb1,Bubeck-survey12}. To handle problems with a large or infinite number of arms, one typically needs side information on similarity between arms. A typical way to model this side information, called \emph{Lipschitz MAB} \cite{LipschitzMAB-stoc08}, is that an algorithm is given a distance function on the arms, and the expected rewards are assumed to satisfy Lipschitz-continuity (or a relaxation thereof) with respect this distance function, e.g.
~\citep{Agrawal-bandits-95,Bobby-nips04,AuerOS/07,LipschitzMAB-stoc08,xbandits-nips08,contextualMAB-colt11}. Most related to this paper is the idea of adaptive discretization which is often used in this setting \citep{LipschitzMAB-stoc08,xbandits-nips08,contextualMAB-colt11}, and particularly the \emph{zooming algorithm} \citep{LipschitzMAB-stoc08,contextualMAB-colt11}. In particular, the general template of our algorithm is similar to the one in the zooming algorithm (but our ``selection rule" and ``zooming rule" are very different, reflecting the lack of a priori known similarity information).

In some settings (including ours), the numerical similarity information required for Lipschitz MAB is not immediately available. For example, in applications to web search and advertising it is natural to assume that an algorithm can only observe a tree-shaped taxonomy on arms~\citep{Kocsis-ecml06,Munos-uai07,yahoo-bandits07,ImplicitMAB-nips11,AdamBull-13}. In particular,~\citet{ImplicitMAB-nips11} and \citet{AdamBull-13} explicitly reconstruct (the relevant parts of) the metric space defined by the taxonomy. In a different direction, \citet{Bubeck-alt11} study a version of Lipschitz MAB where the Lipschitz constant is not known, and essentially recover the performance of \NonAdaptive for this setting.

%\ascomment{I am not discussing my NIPS11 paper in detail. Maybe it is better not to draw attention to it ...}

\vspace{2mm}
\emph{Dynamic pricing} (a.k.a. online posted-price auctions) refers to settings in which a principal interacts with agents that arrive over time and offers each agent a price for a transaction, such as selling or buying an item. The version in which the principal \emph{sells} items has been extensively studied in Operations Research, typically in a Bayesian setting; see \citet{Boer-survey15} for a through literature review. The study of prior-independent, non-parameterized formulations has been initiated in \citet{Blum03} and \citet{Bobby-focs03} and continued by several others \citep{BZ09,DynPricing-ec12,BesbesZeevi-or12,Wang-OR14,BwK-focs13,cBwK-colt14}. Further, \citet{DynProcurement-ec12} and \citet{Krause-www13} studied the version in which the principal \emph{buys} items, or equivalently commissions tasks; we call this version \emph{\taskPricing}. Modulo budget constraints, this is essentially the special case of our setting where in each round a worker is offered the chance to perform a task at a specified price, and can either accept or reject this offer. In particular, the worker's strategic choice is directly observable. More general settings have been studied in \citep{BwK-focs13,cBwK-colt14,AgrawalDevanur-ec14,AgrawalDevanurLi-15}.%
\footnote{\citet{cBwK-colt14} and \citet{AgrawalDevanur-ec14} is concurrent and independent work with respect to the conference publication of this paper, and \citet{AgrawalDevanurLi-15} is subsequent work.} However, all this work (after the initial papers \citep{Blum03,Bobby-focs03}) has focused on models with constraints on the principal's supply or budgets, and does not imply any improved results when specialized to unconstrained settings.

% the long related work section

\section{Conclusions}

Motivated by applications to crowdsourcing markets, we define the \emph{\problem problem}, a multi-round version of the principal-agent model with unobservable strategic decisions. We treat this problem as a multi-armed bandit problem, design an algorithm for this problem, and derive regret bounds which compare favorably to prior work. Our main conceptual contribution, aside from identifying the model, is the adaptive discretization approach that does not rely on Lipschitz-continuity assumptions. We provably improve on the uniform discretization approach from prior work, both in the general case and in some illustrative special cases. These theoretical results are supported by simulations.

We believe that the \problem problem deserves further study, in several directions that we outline below.

\fakeItem[1.] It is not clear whether our provable results can be improved, perhaps using substantially different algorithms and relative to different problem-specific structures. In particular, one needs to establish \emph{lower bounds} in order to argue about optimality; no lower bounds for \problem are currently known.

\fakeItem[2.] Our adaptive discretization approach may be fine-tuned
to improve its performance in practice. In particular, the definition
of the ``index" $I_t(C)$ of a given feasible cell $C$ may be
re-defined in several different ways. First, it can use the
  information from $C$ in a more sophisticated way, similar to the
  more sophisticated indices for the basic $K$-armed bandit problem;
  for example, see \citet{Garivier-colt11}. Second, the index can
  incorporate information from other cells. Third, it can be defined
  in a ``smoother", probabilistic way,  e.g., as in Thompson Sampling \citep{Thompson-1933}.

\OMIT{One possible area of improvement is selecting a feasible cell in a ``smoother" probabilistic way,  e.g., as in Thompson Sampling \citep{Thompson-1933}.}

\fakeItem[3.] Deeper insights into the structure of the (static) principal-agent problem are needed, primarily in order to optimize the choice of $\Xcand$, the set of candidate contracts. The most natural target here is the uniform mesh $\Xcand(\eps)$. To optimize the granularity $\eps$, one needs to upper-bound the discretization error
    $\OPT(\Xcand) - \OPT(\Xcand(\eps))$
in terms of some function $f(\eps)$ such that $f(\eps)\to 0$ as $\eps\to 0$. The first-order open question is to resolve whether this can be done in the general case, or provide a specific example when it cannot. A related open question concerns the effect of increasing the granularity: upper-bound the difference
    $\OPT(\Xcand(\eps)) - \OPT(\Xcand(\eps'))$, $\eps>\eps'>0$,
in terms of some function of $\eps$ and $\eps'$. Further, it is not known whether the optimal mesh of contracts is in fact a uniform mesh.

Also of interest is the effect of restricting our attention to monotone contracts. While we prove that monotone contracts may not be optimal (Appendix~\ref{sec:example-monotone}), the significance of this phenomenon is unclear. One would like to characterize the scenarios when restricting to monotone contracts is alright (in the sense that the best monotone contract is as good, or not much worse, than the best contract), and the scenarios when this restriction results in a significant loss. For the latter scenarios, different algorithms may be needed.

\fakeItem[4.] A much more extensive analysis of special cases is in order. Our general results are difficult to access (which appears to be an inherent property of the general problem), so the most immediate direction for special cases is deriving lucid corollaries from the current regret bounds. In particular, it is desirable to optimize the choice of candidate contracts. Apart from ``massaging" the current results, one can also design improved algorithms and derive specialized lower bounds. Particularly appealing special cases concern supply distributions that are mixtures of a small number of types, and supply distributions that belong to a (simple) parameterized family with unknown parameter.

\vspace{2mm}

Going beyond our current model, a natural direction is to
  incorporate a budget constraint, extending the corresponding results
  on dynamic task pricing. The main difficulty for such settings is
  that a \emph{distribution} over two  contracts may perform much
  better than any fixed contract; see \citet{BwK-focs13} for
  discussion. Effectively, an algorithm needs to optimize over the
  distributions. As a first step, one can use non-adaptive
  discretization in conjunction with the general algorithms for
  bandits with budget constraints (sometimes called ``bandits with knapsacks"
  \citep{BwK-focs13,AgrawalDevanur-ec14}). However, it is not clear
  how to choose an optimal mesh of contracts (as we discussed
  throughout the paper), and this mesh is not likely to be uniform
  (because it is not uniform for the special case of dynamic task
  pricing with a budget \citep{BwK-focs13}). The eventual target in
  this research direction is to marry adaptive discretization and the
  techniques from prior work on ``bandits with knapsacks.''

\begin{small}
\bibliography{bib-abbrv,dynamicPA,bib-bandits,bib-AGT,bib-slivkins,bib-crowdsourcing}
\end{small}

%\newpage
\appendix
%\section*{Appendix}
% Moved these to the order in which they are referenced in the main body.

\section{Monotone contracts may not be optimal}
\label{sec:example-monotone}

In this section we provide an example of a problem instance for which all monotone contracts are suboptimal (at least when restricting attention to only those contracts with non-negative payoffs).  In this example, there are three non-null outcomes (i.e., $m=3$), and two non-null effort levels, ``low'' effort and ``high'' effort, which we denote $e_{\ell}$ and $e_{h}$ respectively.  There is only a single worker type.  Since there is only one type, we drop the subscript when describing the cost function $c$.  We let $c(e_{\ell}) = 0$, and let $c(e_{h})$ be any positive value less than $0.5( v(2) -v(1)).$  If a worker chooses low effort, the outcome is equally likely to be 1 or 3.  If the worker chooses high effort, it is equally likely to be 2 or 3.  It is easy to verify that this type satisfies the FOSD assumption.  Finally, for simplicity, we assume that all workers break ties between high effort and any other effort level in favor of high effort, and that all workers break ties between low effort and the null effort level in favor of low effort.

Let's consider the optimal contract.  Since there is just a single worker type and all workers of this type break ties in the same way, we can consider separately the best contract that would make all workers choose the null effort level, the best contract that would make all workers choose low effort, and the best contract that would make all workers choose high effort, and compare the requester's expected value for each.

Since $c(e_{\ell}) = 0$ and workers break ties between low effort and null effort in favor of low effort, there is no contract that would cause workers to choose null effort; workers always prefer low effort to null effort.

It is easy to see that the best contract (in terms of requester expected value) that would make workers choose low effort would set $x(1) = x(3) = 0$ and $x(2)$ sufficiently low that workers would not be enticed to choose high effort; setting $x(2) = 0$ is sufficient.  In this case, the expected value of the requester would be $0.5 (v(1) + v(3))$.

Now let's consider contracts that cause workers to choose high effort.  If a worker chooses high effort, the expected value to the requester is
\begin{equation}
0.5 (v(2) - x(2) + v(3) - x(3)).
\label{eqn:reqexpv}
\end{equation}
Workers will choose high effort if and only if
\[
0.5 ( x(1) + x(3) ) \leq 0.5 ( x(2) + x(3) ) - c(e_h)
\]
or
\begin{equation}
0.5 x(1)  \leq 0.5 x(2)  - c(e_h).
\label{eqn:choosecond}
\end{equation}
So to find the contract that maximizes the requester's expected value when workers choose high effort, we want to maximize Equation~\ref{eqn:reqexpv}  subject to the constraint in Equation~\ref{eqn:choosecond}.  Since $x(3)$ doesn't appear in Equation~\ref{eqn:choosecond}, we can set it to $0$ to maximize Equation~\ref{eqn:reqexpv}.  Since $x(1)$ does not appear in Equation~\ref{eqn:reqexpv}, we can set $x(1) = 0$ to make Equation~\ref{eqn:choosecond} as easy as possible to satisfy. We can then see that the optimal occurs when $x(2) = 2 c(e_h)$.

Plugging this contact $x$ into Equation~\ref{eqn:reqexpv}, the expected utility in this case is $0.5 ( v(2) + v(3) ) - c(e_h)$. Since we assumed that $c(e_h) < 0.5 ( v(2) -v(1) ))$, this is strictly preferable to the constant 0 contract, and is in fact the unique optimal contract.   Since $x(2) > x(3)$, the unique optimal contract is not monotonic.

% commenting out this section
\ignore{
\subsection{Optimality of Linear Contracts in Moral Hazard Setting}
\label{sec:linear_optimal}

This section is mostly a summary of the existing results.
Below are the notations we use.

\begin{itemize}
   \item The agent chooses an effort level $e$.
   \item The agent has a cost function $c(e)$ for exerting effort level $e$.
   \item Given effort level $e$, the principal obtains profit $\pi$, which
         is conditional distributed in pdf $p(\pi|e)$.
   \item The principal sets a contract $w(\pi)$, which is the payment he/she
         gives to the agent when obtaining profit $\pi$.
\end{itemize}

We assume the agent is risk neutral, i.e., the agent's utility $v(w(\pi))$
for obtaining payment $w(\pi)$ is equal to $w(\pi)$ (without considering the
uncertainty of the payments.)  Given a contract $w(\pi)$ and  agent's effort
level $e$, the agent's payoff is
\[
\int_\pi w(\pi) p(\pi|e) d\pi - c(e),
\]
and the principal's payoff is
\[
\int_\pi (\pi - w(\pi)) p(\pi|e) d\pi.
\]

In the principal agent problem, the principal aims to design the contract
maximizing his/her own payoff.
}  % ends ignore of section on optimality of linear contracts

% commenting out this section
\ignore{
\subsection{Observable Effort}

Let $u$ be the minimum payoff for the agent to participate,
\ignore{the optimal
contract for the principal would be the optimal of the following
constrained maximization problem.
\begin{align*}
   \max_{e, w(\pi)} & \int_\pi (\pi - w(\pi)) p(\pi|e) d\pi \\
      \mbox{s.t.} & \int_\pi w(\pi) p(\pi|e) d\pi - c(e) \geq u
\end{align*}
}
the optimal contract can be solved in two stages. First, for any given
effort level $e$, the optimal contract $w(\pi)$ for $e$ is the optimal
solution of the following
\begin{align*}
   \max_{w(\pi)} & \int_\pi (\pi - w(\pi)) p(\pi|e) d\pi \\
   \mbox{s.t.} & \int_\pi w(\pi) p(\pi|e) d\pi - c(e) \geq u
\end{align*}

This problem can be solved by checking the
first order condition of the Lagrangian. We can then find the optimal $e^*$
of the maximization problem. So the optimal contract with observable efforts
is to ask the worker to exert effort level $e^*$ and pay him/her the
corresponding amount if he/she does so.

Note that in the optimal solution, the equality of the constraint holds,
otherwise, we can decrease the payment to increase the objective without
violating the constraint. Therefore,
\[
   \int_\pi w^*(\pi) p(\pi|e^*) d\pi - c(e^*) = u
\]

Insert this back to the objective of the maximization problem, we know the
the optimal effort level $e^*$ when the effort is observable
solves the following maximization problem.
\begin{align}
   \max_{e} \int_\pi \pi p(\pi|e) d\pi -u - c(e)
\label{eqn:eqn1}
\end{align}
}  % ends ignore of section on observable effort

% commenting out this section
\ignore{
\subsection{Unobservable Effort}

Below we show that the simple linear contract $w(\pi) = \pi + K$ is
an optimal contract when the agent effort is not observable.

Consider a linear contract $w(\pi) = \pi + K $, where $K$ is some constant.
If the agent accepts this contract, he/she will exert
effort level $e$ to maximize his/her own payoff
\begin{align}
   \int_\pi w(\pi) p(\pi|e) d\pi - c(e) = \int_\pi \pi p(\pi|e) d\pi + K -c(e)
\label{eqn:eqn2}
\end{align}

Comparing Equation~\ref{eqn:eqn1} and~\ref{eqn:eqn2}, we know that the
optimal effort level $e^*$ when effort is observable also maximizes
equation~\ref{eqn:eqn2}. Therefore, if the agent accepts this contract,
he/she will exert the optimal effort as if the effort level is observable.

The agent would accept the contract as long as
(recall that $u$ is the minimum payoff the agent is willing to work for)
\[
   \int_\pi \pi p(\pi|e^*) d\pi + K - c(e^*) \geq u.
\]

Therefore, if we set the contract $w(\pi) = \pi + K^*$,
where $K^* = c(e^*) + u - \int_\pi \pi p(\pi|e^*) d\pi$. We can ensure that
the agent exerts effort level $e^*$ and also the principal can obtain the
same optimal payoff as if the effort is observable
(note that the principal's payoff is $-K^*$ in this case.)
} % ends ignore of section on unobservable effort  % monotone contracts may not be optimal

\end{document}